\documentclass{article} 
\usepackage{arxiv}  % DO NOT CHANGE THIS
\usepackage[utf8]{inputenc} % allow utf-8 input
\usepackage[T1]{fontenc}    % use 8-bit T1 fonts
\usepackage{hyperref}       % hyperlinks
\usepackage{url}            % simple URL typesetting
\usepackage{booktabs}       % professional-quality tables
\usepackage{amsfonts}       % blackboard math symbols
\usepackage{nicefrac}       % compact symbols for 1/2, etc.
\usepackage{microtype}      % microtypography
\usepackage{lipsum}
\usepackage{xparse}
\usepackage{times}
\usepackage{soul}
\usepackage{url}
\usepackage[small]{caption}
\usepackage{graphicx}
\usepackage{amsmath}
\usepackage{amsthm}
\usepackage{booktabs}
\usepackage{algorithm}
\usepackage{algorithmic}
\usepackage{todonotes}
\usepackage{mathtools}
\usepackage{multirow}
\usepackage{macros}

\usepackage{amssymb}
\usepackage{natbib}
\usepackage{hyperref}
\usepackage{nicefrac, xfrac}

\usepackage{thmtools}
\usepackage{thm-restate}
\usepackage[capitalize,noabbrev]{cleveref}

\theoremstyle{plain}
\newtheorem{theorem}{Theorem}[section]

\newtheorem{lemma}{Lemma}[section]

\theoremstyle{definition}
\newtheorem{definition}[theorem]{Definition}

\theoremstyle{remark}

\usepackage{newfloat}
\usepackage{listings}
\DeclareCaptionStyle{ruled}{labelfont=normalfont,labelsep=colon,strut=off} % DO NOT CHANGE THIS
\floatstyle{ruled}
\newfloat{listing}{tb}{lst}{}
\floatname{listing}{Listing}
\title{Optimal Rates and Efficient Algorithms for Online Bayesian Persuasion}
\author{
	Martino Bernasconi\\
	Politecnico di Milano\\
	\texttt{martino.bernasconideluca@polimi.it}
	\And
	Matteo Castiglioni\\
	Politecnico di Milano\\
	\texttt{matteo.castiglioni@polimi.it}
	\And
	Andrea Celli\\
	Bocconi University \\
	\texttt{andrea.celli2@unibocconi.it}
	\And
	Alberto Marchesi\\
	Politecnico di Milano\\
	\texttt{alberto.marchesi@polimi.it}
	\And
	Nicola Gatti\\
	Politecnico di Milano\\
	\texttt{nicola.gatti@polimi.it}
	\And
	Francesco Trovò\\
	Politecnico di Milano\\
	\texttt{francesco1.trovo@polimi.it}
}

\begin{document}

\maketitle

\begin{abstract}
	Bayesian persuasion studies how an informed sender should influence beliefs of rational receivers who take decisions through Bayesian updating of a common prior. We focus on the \emph{online Bayesian persuasion} framework, in which the sender repeatedly faces one or more receivers with unknown and adversarially selected types.
	First, we show how to obtain a tight $\tilde O(T^{1/2})$ regret bound in the case in which the sender faces a single receiver and has partial feedback, improving over the best previously known bound of $\tilde O(T^{4/5})$. Then, we provide the first no-regret guarantees for the multi-receiver setting under partial feedback.
	Finally, we show how to design no-regret algorithms with polynomial per-iteration running time by exploiting \emph{type reporting}, thereby circumventing known intractability results on online Bayesian persuasion. We provide efficient algorithms guaranteeing a $ O(T^{1/2})$ regret upper bound  both in the single- and multi-receiver scenario when type reporting is allowed.
\end{abstract}

% \section{TODO}

% Multi-receiver con due azioni! e va definito $\phi(R)$ dove $R$ è il set di receivers che giocano $a_1$.
\section{Introduction}\label{sec:intro}

The \emph{Bayesian persuasion} framework, introduced by \citet{kamenica2011bayesian}, is an economic model which helps to explain how individuals make decisions based on the information they receive from others, and how this information can be used to influence their behavior. This model is particularly useful for understanding strategic interactions in situations where individuals have different levels of information or expertise. The framework already found application in domains such as advertising~\citep{badanidiyuru2018targeting,emek2014signaling,bro2012send,Castiglioni2022Posted,bacchio2022}, voting~\citep{alonso2016persuading,castiglioni2019persuading,cheng2015mixture,CastiglioniPersuading2021}, routing~\citep{bhaskar2016hardness,vasserman2015implementing,castiglioni2021signaling}, security~\citep{rabinovich2015information,xu2016signaling}, and in incentivized exploration in multi-armed bandits \citep{kremer2014implementing,cohen2019optimal,mansour2016bayesian,sellke2021price,mansour2022bayesian}.

In the simplest instantiation of the model, there are a sender and a receiver with a common prior over a finite set of states of nature. The sender publicly commits to a \emph{signaling scheme}, which is a randomized mapping from states of nature to signals being sent to the receiver. Then, the sender observes the realized state of nature, and they send a signal to the receiver following the signaling scheme. The receiver observes the signal, computes their posterior distribution over states, and selects an action maximizing their expected utility. The sender and the receiver obtain a payoff which is a function of the receiver's action, and of the realized state of nature. An optimal signaling scheme for the sender is one maximizing their expected utility.

The study of Bayesian persuasion from a computational perspective was initiated by \citet{dughmi2016algorithmic}, and the original model was later extended to more complex settings such as games with multiple receivers (see, \eg \citep{dughmi2017algorithmicExternalities,bhaskar2016hardness,xu2020tractability}). 
A key question that has emerged is whether computational techniques can be used to ease some of the assumptions made in the original model by \citet{kamenica2011bayesian}. Two main lines of research have emerged from this question: one is aimed at developing robust algorithms that can bypass the common-prior assumption \citep{camara2020mechanisms,zu2021learning,bernasconisequential}, and the other is focused on the robustness of persuasion when the sender is unaware of the receiver's goals \citep{castiglioni2020online,babichenko2021regret}.

This work follows the second perspective, and studies the \emph{online Bayesian persuasion} framework introduced by \citet{castiglioni2020online}. In this framework, the sender repeatedly faces a receiver whose type is unknown and chosen adversarially at each round from a finite set of possible types. This framework encompasses the problem of learning in repeated Stackelberg games \citep{letchford2009learning,blum2014learning,marecki2012playing,balcan2015commitment}.

%\vspace{2mm} \noindent{\bf Contributions.} 
\paragraph{Contributions}
We start by describing a general no-regret algorithm for online-learning against an oblivious adversary with a \emph{finite} number of possible loss functions. We use this algorithm to provide a tight $\tilde O(T^{1/2})$ regret upper bound in the setting with one receiver and partial feedback, improving over the $\tilde O(T^{4/5})$ rate by \citet{castiglioni2020online}. This result also improves the best known bound of $\tilde O(T^{2/3})$ for online-learning in repeated Stackelberg games by \citet{balcan2015commitment}. Then, we show that our general framework can be applied to obtain the first no-regret guarantees in the multi-receiver setting by \citet{castiglioni2021multi} under partial feedback. In particular, we provide a tight $\tilde O(T^{1/2})$ regret bound under the assumption the set of possible type profiles of the receivers is known beforehand by the sender.
In each of these settings, our no-regret algorithms may suffer from exponential per-iteration running time, as expected from known hardness results for the online Bayesian persuasion settings \cite{castiglioni2020online}. In the last part of the paper, 
we provide the first no-regret algorithms for online Bayesian persuasion with guaranteed polynomial per-iteration running time. We do that by considering the \emph{type reporting} framework by~\citet{DBLP:conf/atal/CastiglioniM022}, where the sender can commit to a \emph{menu} of signaling schemes, and then let the receivers choose their preferred signaling scheme depending on their private types. In such a setting, we provide a $O(T^{1/2})$ regret upper bound for the single-receiver setting. Moreover, by designing a general algorithm based on FTRL, we shot that it is possible to achieve the same rate of convergence with polynomial-time per-iteration time complexity also in the multi-receiver setting, when receivers have binary actions and the utility of the sender is specified by a supermodular or anonymous function.

\section{Preliminaries}\label{sec:preliminaries}

Vectors are denoted by bold symbols. Given a vector $\vx$, we let $x_i$ be its $i$-th component. The set $\{1,2,\ldots,n\}$ of the first $n$ natural numbers is compactly denoted as $[n]$. Moreover, given a discrete set $\cX$, we denote by $\Delta_{\cX}$ the $|\cX|$-simplex, while, given a set $\cY$, $\textnormal{int}(\cY)$ is the {\em interior} of $\cY$.

In the following, we formally describe the \emph{online Bayesian persuasion} (OBP) framework originally introduced by~\citet{castiglioni2021multi}.
Such a framework models a repeated interaction between a \emph{sender} and multiple \emph{receivers}.

We denote by $\cR \defeq [n]$ a finite set of $n$ receivers. Each receiver $r\in\cR$ has a finite set $\cK_r$ of $m$ different types,
%$cK\defeq \{k_{r, i} \}_{i=1}^{m}$
%Each receiver $r\in\cR$ has a finite set $\A_r \defeq \{ a_{r,i} \}_{i=1}^\ell$ of $\ell$ available \emph{actions}.
and a finite set $\cA_r$ of available \emph{actions}.
We let $\cK \defeq \bigtimes_{r \in \cR} \cK_r$ be the set of \emph{type profiles}, \ie vectors $\vk \in \cK$ defining a type $k_r \in \cK_r$ for each receiver $r \in \cR$.
%
%Similarly, we let $\cA \defeq  \bigtimes_{r \in \cR} \cA_r$ be the set of \emph{action profiles} $\va \in \A$ specifying an action $a_r \in \cA_r$ for each $r \in \cR$.
Similarly, we let $\cA \defeq  A^n$ be the set of \emph{action profiles} $\va \in \A$ specifying an action $a_r \in A$ for each receiver $r \in \cR$.\footnote{We assume that all the receivers have the same action set $A$. This comes w.l.o.g. as it is always possible to add fictitious actions to the receivers whenever the assumptions does \emph{not} hold.}

The payoffs of both the sender and the receivers depend on a random \emph{state of nature}, which is drawn from a finite set $\Theta$ of $d$ possible states according to a commonly-known \emph{prior} probability distribution $\vmu  \in\textnormal{int}(\Delta_\Theta)$.
The sender's payoffs also depend on the actions selected by the receivers, as defined by the function $\us: \cA \times \Theta \to [0,1]$.
Moreover, as it is customary in the literature (see, \eg \citep{dughmi2017algorithmicExternalities}), we assume that there are \emph{no inter-agent externalities}, which means that the payoffs of a receiver only depend on the action played by them, and \emph{not} on those played by other receivers.
Formally, a receiver $r \in \cR$ of type $k \in \cK_r$ is characterized by a payoff function $\ur[r][k]: \cA_r \times\Theta\to [0,1]$.

%\andrea{payoff matrix $\mU^{k}\in [0,1]^{\cA_r\times\Theta}$ e $\mU^\sender\in[0,1]^{\cA\times\Theta}$ ?}

As in the classical Bayesian persuasion framework by~\citet{kamenica2011bayesian}, the sender gets to know the realized state of nature $\theta \sim \vmu$, and they have the ability to strategically disclose (part of) such information to the receivers, in order to maximize their own utility.
This is achieved by committing beforehand to a \emph{signaling scheme}, which is a randomized mapping from states of nature to signals being sent to the receivers. Formally, let $\cS \defeq \bigtimes_{r \in \cR} \cS_r$ be the finite set of \emph{signal profiles}, \ie the set of vectors $\vs \in \cS$ defining a signal $s_r \in \cS_r$ for each receiver $r \in \cR$.\footnote{In this work, we focus on \emph{private} signaling, where the sender has the ability to privately communicate a signal to each receiver.} 
Then, a signaling scheme is a mapping $\phi : \Theta \to \Delta_\cS$. 
We denote by $\phi_\theta(\vs)$ the probability of sending the signals in $\vs \in \cS$ when the state of nature is $\theta \in \Theta$.
%
%For ease of notation, we denote by $\phi_\theta \in \Delta_\cS$ the probability distribution over signal profiles corresponding to state $\theta \in \Theta$, while we let $\phi_\theta[\vs]$ be the probability of sending $\vs \in \cS$ under state $\theta$.
%
Moreover, given a signaling scheme $\phi$, we define the resulting \emph{marginal} signaling scheme for a receiver $r \in \cR$ as $\phi^r: \Theta \to \Delta_{\cS_r}$.
Formally, for every $\theta \in \Theta$, the marginal signaling scheme $\phi^r$ defines the distribution over receiver $r$'s signals that is induced by $\phi$, which assigns probability 
\begin{equation}\label{eq:marginal}
	\phi^r_\theta (s') \defeq \sum_{\vs \in \cS: s_r = s'} \phi_\theta(\vs) \,\,\text{\normalfont to each } s' \in \cS_r.
\end{equation}
%
%the distribution $\phi^r_\theta  \in \Delta_{\cS_r}$ assigns probability $\phi^r_\theta [s']= \sum_{\vs \in \cS: s_r = s'} \phi_\theta[\vs]$ to each receiver's (private) signal $s' \in \cS_r$.

% The multi-receiver framework that we consider was introduced by \citet{castiglioni2021multi}.
%
The repeated interaction between the sender and the receivers goes on as follows.
At each round $t\in [T]$, the sender commits to a signaling scheme $\phi_t$ (\ie $\phi_t$ is publicly known), and, subsequently, they observe the realized state of nature $\theta \sim \vmu$.
Then, the sender draws a signal profile $\vs\sim\phi_{t,\theta}$ and communicates to each receiver $r \in \cR$ (whose type is unknown to the sender) their own private signal $s_r$.
After observing the signal, each receiver $r \in \cR$ updates their prior belief $\vmu$ according to Bayes rule, and, then, they select an action maximizing their expected utility.

The \emph{posterior} $\vxi^{s_r} \in \Delta_\Theta$ computed by a receiver $r \in \cR$ after observing a signal $s_r \in \cS_r$ under signaling scheme $\phi$ is a probability distribution over states such that 
\[
	\xi_{\theta}^{s_r} \defeq \frac{\mu_\theta\, \phi_\theta^r(s_r)}{ \sum_{\theta'\in\Theta} \mu_{\theta^\prime}\phi_{\theta^\prime}^r(s_r) }\hspace{.5cm}\text{\normalfont for every $\theta\in\Theta$}.
	\footnote{For ease of notation, we omit the dependence of the posterior distribution $\vxi^{s_r}$ from the signaling scheme $\phi$ and the receiver $r$, as these will be clear from context.}
\]
%
% computed by receiver $r \in \cR$ after observing signals $s_r \in \cS_r$ at round $t$ is defined by the probabilities
%at time $t$ for receiver $r$ having observed signal $s_r$, is computed as 
%\[
%p^r_\theta(\vmu,\phi_t,s_r) \defeq \frac{\mu_\theta\, \phi^r_{t,\theta}[s_r]}{ \sum_{\theta'\in\Theta} \phi^r_{t,\theta'}[s_r] }\hspace{.5cm}\text{\normalfont for every $\theta\in\Theta$},
%\]
%where $\phi^r_t$ is the marginal signaling scheme induced by $\phi_t$ for receiver $r$.
%
% In the following, we let $\Xi\defeq \Delta_\Theta$ be the set of possible posteriors, and we compactly denote the posterior of receiver $r$ at round $t$ by the vector $\vxi_t^r \in \Xi$ whose components are defined as $\xi_{t,\theta}^r \defeq p_r(\vmu,\phi_t,s_r)$ for all $\theta \in \Theta$.
%
Given a posterior $\vxi\in \Delta_\Theta$, the set of \emph{best-response actions} of a receiver $r \in \cR$ of type $k \in\cK_r$ is defined as follows: 
%\[
%\cB^k_{\vxi}\defeq \argmax_{a\in\cA_r} \mU^k_a\, \vxi.
%\]
%
\[
	\cB^{r,k}_{\vxi}\defeq \argmax_{a\in\cA_r} \sum_{\theta \in \Theta} \xi_{\theta} \, \ur[r][k](a,\theta).
\]
%
% Thus, if receiver $r \in \cR$ has type $k \in \cK_r$ at round $t \in [T]$, then they select an action in $\cB_{\vxi_t^r}^{r,k}$.
%
%Given a signaling scheme $\phi$, each signal $\vs\in\cS$ induces a posterior for each reveicer. We denote by $\vw_\phi \in\Delta_{\Xi}$ the probability with which $\phi$ induces each possible tuple of posterior beliefs. Formally, for each $\uxi=(\vxi_1,\ldots,\vxi_n)\in\Xi^n$, 
%\[
%w_{\phi,\uxi}=\sum_{\vs\in\cS:  \uxi_r = \vxi^r_{\phi,s_r} \,\forall r\in\cR } \vmu^\top \phi[\vs].
%\]
%We denote the set of posteriors induced with strictly positive probability by $\phi$ as 
%\[
%\supp(\phi)\defeq \mleft\{ \uxi\in\Xi^n: w_{\phi,\uxi}>0 \mright\}.
%\]
%
Moreover, assuming receivers break ties in favor of the sender, the sender's expected utility for selecting a signaling scheme $\phi$ given a receivers' type profile $\vk \in \cK$ is
\[
\us(\phi,\vk) \defeq \sum_{\vs\in\cS} \mleft(\argmax_{\va \in \bigtimes\limits_{r \in \cR} \cB_{\vxi^{s_r}}^{r,k_r}}  \sum_{\theta \in \Theta} \mu_\theta \phi_\theta(\vs) \us(\va, \theta)\mright).
\]
%
%\[
%	\us(\phi,\vk) \defeq \sum_{\theta\in\Theta} \mu_\theta \sum_{\vs\in\cS} \phi_\theta[\vs] \argmax_{\substack{\va\in\cA: a_r\times \cB^{k_r}_{\vxi},\\\vxi=p_r(\vmu,\phi,s_r)}}  \mU^\sender_{\va,\theta}
%\]

We focus on the problem of computing a sequence $\left\{ \phi_t \right\}_{t \in [T]}$ of signaling schemes which can be employed by the sender so as to maximize their utility.
We assume that the sequence of receivers' type profiles $\left\{ \vk_t \right\}_{t \in [T]}$, with $\vk_t \in \cK$, is selected by an oblivious adversary.
At each round $t \in [T]$ of the repeated interaction, the sender gets a payoff $\us(\phi_t, \vk_t)$ and receives some feedbacks about receivers' types.
In the \emph{full feedback} setting, the sender gets to know the receivers' type profile $\vk_t$, while in the \emph{partial feedback} setting the sender only observes the action profile $\va_t \in \cA$ played by the receivers at round $t$.
We measure the performance of the sender by using the regret up to round $T$ with respect to the best fixed signaling scheme in hindsight:
\begin{equation*}
R_T\defeq \max_{\phi} \sum_{t=1}^T \us(\phi,\vk_t)- \sum_{t=1}^T \E\mleft[\us(\phi_t,\vk_t) \mright] ,
\end{equation*}
%\begin{equation*}
%R_T\defeq \max_{\phi}\mleft\{ \sum_{t=1}^T \mleft(\us(\phi,\vk_t)- \E\mleft[\us(\phi_t,\vk_t) \mright] \mright) \mright\},
%\end{equation*}
where the expectation is on the possible randomness of the algorithm.\footnote{This notion of regret is also known as \emph{Stackelberg regret} \citep{balcan2015commitment,chen2020learning}.}
Ideally, we would like an algorithm that generates a sequence $\{\phi_t\}_{t \in [T]}$ with the following properties: (i) the regret is polynomial in the size of the problem instance, \ie it is $\poly(n, m, d, |A|)$, and goes to zero as $T\to\infty$; and (ii) the per-round running time is $\poly(t, n, m, d, |A|)$.
%
%
%At the end of each iteration $t \in [T]$, the sender gets payoff $\us(\phi_t, \vk_t)$, and observes the action profile $\va_t$, that is, the sender has to take decisions under  \emph{bandit feedback}.\footnote{The \emph{full-information feedback} setting would correspond to observing the full vector of types $\vk_t$ after each step $t$.}
%%
%We focus on the problem of computing a sequence $\{\phi^t\}_{t \in [T]}$ of signaling schemes maximizing the sender's expected utility when the sequence of receivers' types $ \{\kvec^t\}_{t\in [T]}$, with $\kvec^t \in\K$, is selected by an oblivious adversary.
%%
%We measure the performance of the sender using the regret up to time $T$ with respect to the best fixed signaling scheme in hindsight:
%\begin{equation*}
%	R_T\defeq \max_{\phi}\mleft\{ \sum_{t=1}^T \mleft(\us(\phi,\vk_t)- \E\mleft[\us(\phi_t,\vk_t) \mright] \mright) \mright\},
%\end{equation*}
%where the expectation is on the randomness of the online algorithm.\footnote{\andrea{menzionare Stackelberg regret}.}
%%
%Ideally, we would like to find an algorithm that generates a sequence of signaling schemes $\{\phi^t\}_{t \in [T]}$ with the following properties: (i) the regret is polynomial in the size of the problem instance, \ie $\poly(n, m, d)$, and goes to zero as $T\to\infty$; (ii) the per-round running time is $\poly(t, n, m, d)$.

\section{Online Learning Against Adversaries with a Finite Number of Losses}\label{sec:smallLosses}

%\ma{Notation clash: $d$ è sia l'azione dell'avversario in questa sezione e il numero di stati di natura nel resto}

We start by introducing a general framework that will be crucial in proving some of our main results in the rest of the paper.
In particular, we propose a no-regret algorithm for a general online learning problem in which the agent's decisions are only evaluated in terms of $D$ possible, adversarially-selected loss functions.
The algorithm that we propose attains a $\tilde O(\sqrt{T})$ regret bound, which is independent of the size of the decision space of the agent and only depends polynomially on the number of possible losses $D$.
%
%\ma{In this Section we discuss the case of adversarial online learning in the case we play a game against an adversary which has only $D$ number of losses. This introduce a general framework for dealing with Online Learning problems against an the adversary that has a finite action space. The main interest of this results is that we can obtain $\sqrt{T}$ regret bounds which are independent from the size of the action space of the agent, and only depend polynomially on $D$.}
%
%Differently from~\cite{daskalakis2022learning} we use a reduction of this problem to the problem of online learning with linear losses.

In the online learning problem that we consider in this section, at each round $t \in [T]$, an agent takes a decision $\vx_t$ from a set $\cX \subseteq \mathbb{R}^M$, and, then, an adversary selects an element $d_t$ from a finite set $\cD$ of $D \coloneqq |\cD|$ elements.
Then, the loss suffered by the agent is $L_{d_t}(\vx_t)$, where functions $L_{d} : \cX \to [0,1]$ are loss functions indexed by the elements $d \in \cD$.
%
%the agent choose a distribution $\xvec^t$ over the set $\K\in\R^M$, the adversary picks an action $d^t$ from a finite set $\D$ of cardinality $D$. Then an action $a^t\in\K$ of the agent is sampled from the distribution $\xvec^t$ and the agent suffers a loss $L(a^t,d^t)$.
%
Thus, the performance of the agent over the $T$ rounds is evaluated by means of the regret $R_T \coloneqq \sum_{t=1}^T\Exp [L_{d_t}(\vx_t)]-\min_{\vx \in\cX}\sum_{t=1}^T L_{d_t}(\vx)$, where the expectation is with respect to the (possible) randomization that the agent adopts in choosing $\vx_t$.

Next, we introduce a general no-regret algorithm that works by exploiting the linear structure of the online learning problem described above.
In order to do so, we introduce a vector-valued function $\vnu:\cX\to\R^D$ defined as $\vnu(\vx)\coloneqq [L_d(\vx)]_{d \in \cD}$ for all $\vx \in \cX$.
By observing that $L_d(\vx) = \vnu(\vx)^\top \onevec_{d}$, where $\onevec_{d} \in \left\{0,1 \right\}^D$ is a vector whose $d$-th component is the only one that is different from zero, we can cast the online learning problem as a new one with linear losses defined over the decision space $\vnu(\cX)$. 
Since $\vnu(\cX)$ may \emph{not} be convex, the algorithm employs a regret minimizer $\fR$ working on the convex hull $\co \vnu(\cX)$.\footnote{In order to see that taking the convex hull is necessary, let $\cX$ be the unit sphere in $\mathbb{R}^2$ and $L_d(\vx)=\|\vx\|_2^{2d}$ for $d \in \D=\{0.5,1\}$. Then, it is easy to verify that $\vnu(\cX)=\{(x,\sqrt{x}) : x\in[0,1]\}$, which is \emph{not} a convex set.}
This is possible since, instead of playing an $\zvec\in\co \vnu(\cX)$, the algorithm can replace it by a suitable randomization of $D+1$ points in $\vnu(\cX)$, which is guaranteed to exist by the Carathéodory's theorem. 
See Algorithm~\ref{alg:generalLinear_new} for the detailed procedure, where we denote by $\vnu^\dagger$ the inverse map of $\vnu$.
Notice that, provided that a suitable regret minimizer $\fR$ is instantiated, the algorithm works both in the full feedback setting, where the agent observes $d_t$, and in the bandit feedback one, in which they only observe $L_{d_t}(\vx_t)$.
%
%
%\ma{
%To exploit the linear structure of the problem we define the function $\vnu:\K\to\R^D$ defined as $\vnu(a):=[L(a,d_i)]_{i\in[D]}$. Observing that $L(a,d)=\vnu(a)^\top \onevec_{d}$, where $\onevec_{d}[d_i]:=[\mathbb{I}_{d_i=d}]_{i\in[D]}$, allows the loss to be written linearly in the space $\vnu(\K)$. However the space $\vnu(\K)$ is not guaranteed to be convex, and thus we have to work on randomized algorithms that works over the set $\co\vnu(\K)$, and utilize an oracle that returns for every point $\tilde\zvec\in\co \vnu(\K)$ a discrete randomization over $D+1$ points in $\vnu(\K)$.\footnote{\ma{To see that the convexification is necessary take $\K$ as the unit sphere in $\mathbb{R}^2$ and
%$L(a,d)=\|a\|_2^{2d}$ for $\D=\{0.5,1\}$. Then it easy to verify that $\vnu(\K)=\{(a,\sqrt{a}),a\in[0,1]\}$, which is not convex.}} These points are guaranteed to exists by the Carathodory's Theorem.
%}
%
%\ma{The expected regret is then defined as follows:
%%
%\[
%R_T:=\sum\limits_{t=1}^T\Exp_{a^t\sim\xvec^t}[L(a^t,d^t)]-\min\limits_{a\in\K}\sum\limits_{t=1}^TL(a,d^t).
%\]
%%
%}

\begin{algorithm}[!htp]\caption{\textsc{No-regret Algorithm}}\label{alg:generalLinear_new}
	\begin{algorithmic}[1]
		\REQUIRE Regret minimizer $\fR$ for the set $\co \vnu(\cX)$ and linear losses; Inverse mapping $\vnu^\dagger$
		\STATE Initialize regret minimizer $\fR$
		\FOR{$t=1, \ldots, T$}
		\STATE $\co \vnu(\cX) \ni \zvec_{t}\gets \fR.\textsc{Reccomend}() $
		%\STATE Find points $\zvec^{t+1}_1,\ldots,\zvec^{t+1}_n\in\nu(\K)$ and distribution $\lambdavec\in\Delta{^n}$ s.t. $\sum_{i=1}^n\lambda_i\zvec^{t+1}_i=\tilde\zvec^{t+1}$
		\STATE $\left\{ \left(\zvec_t^i,\lambda^i_t\right) \right\}_{i\in[D+1]}\gets\textsc{Carath\'eodory}({\zvec}_t,\vnu(\cX))$
		\STATE Draw $j \in [D+1]$ with probabilities $\lambda_t^1, \ldots, \lambda_t^{D+1}$
		\STATE Play $\vx_t \gets \vnu^\dagger(\zvec_{t}^j)$
		\STATE {\setlength{\tabcolsep}{0pt}
			\newlength\q
			\setlength\q{\dimexpr .27\textwidth}\noindent\begin{tabular}{p{\q}r}
				Observe $d_t \in\D$ \hfil &\Comment{\textcolor{gray}{\textnormal{Full feedback}}} \\
				Observe $L_{d_t}(\vx_t)$ \hfil &\Comment{\textcolor{gray}{\textnormal{Bandit feedback}}}
		\end{tabular}}
		\STATE {\setlength{\tabcolsep}{0pt}
			\setlength\q{\dimexpr .27\textwidth}\noindent\begin{tabular}{p{\q}r}
				$\fR.\textsc{ObserveLoss}(L_{d_t})$ \hfil &\Comment{\textcolor{gray}{\textnormal{Full feedback}}} \\
				$\fR.\textsc{ObserveLoss}(L_{d_t}(\vx_t))$ \hfil &\Comment{\textcolor{gray}{\textnormal{Bandit feedback}}}
		\end{tabular}}
		%		\STATE Observe $d_t \in\D$ \hfill \Comment{\textcolor{gray}{\textnormal{$L_{d_t}(\vx_t)$ with~bandit~feedback}}}
		%		\STATE $\fR.\textsc{Update}(L_{d_t})$ \hfill  \Comment{\textcolor{gray}{\textnormal{$\fR.\textsc{Update}(L_{d_t}(\vx_t))$ bandit~fb.}}}
		\ENDFOR
	\end{algorithmic}
\end{algorithm}
%
%\ma{Old Algorithm}
%\begin{algorithm}[H]\caption{\textsc{no-regret Algorithm}}\label{alg:generalLinear}
%	\begin{algorithmic}[1]
%		\REQUIRE Regret minimizer $\A$ for the set $\co \vnu(\K)$ and linear losses and inverse oracle $\vnu^\dagger$
%		\STATE Initialize Regret minimizer $\A$
%		\FOR{ $t=1$ to $T$}
%		\STATE $\tilde\zvec^{t}\gets \A.\textsc{GetStrategy}()\in \co \vnu(\K)$
%		%\STATE Find points $\zvec^{t+1}_1,\ldots,\zvec^{t+1}_n\in\nu(\K)$ and distribution $\lambdavec\in\Delta{^n}$ s.t. $\sum_{i=1}^n\lambda_i\zvec^{t+1}_i=\tilde\zvec^{t+1}$
%		\STATE $\left(\zvec^t_i,\lambda_i^t\right)_{i\in[D+1]}\gets\textsc{Carathodory}(\tilde{\zvec}^t,\vnu(\K))$
%		\STATE Sample $j$ with probability $\lambda^t_j$
%		\STATE Adversary picks $d^{t}\in\D$
%		\STATE Play $a^t\sim\xvec_j^{t}=\nu^\dagger(\zvec^{t}_j)$
%		\STATE $\A.\textsc{Update}(L(\xvec_j^{t}, d^{t}))$\ma{Da discutere che feedback vede l'algoritmo} %\Comment{Update Regret Minimizer}
%		\ENDFOR
%	\end{algorithmic}
%\end{algorithm}

The following theorem bounds the regret of Algorithm~\ref{alg:generalLinear_new}:
%
%\begin{theorem}
%	Algorithm~\ref{alg:generalLinear_new} guarantees a cumulative regret $R_T \leq O(\poly(D)\sqrt{T})$ with either full or bandit feedback.
%	%
%	% Both under full and partial feedback there is an algorithm which achieves regret $O(\poly(D)\sqrt{T})$.
%\end{theorem}
\begin{restatable}{theorem}{generalLinear}\label{th:generalLinear}
	%Algorithm~\ref{alg:generalLinear_new} guarantees a cumulative regret $R_T \leq O(\poly(D)\sqrt{T})$ with either full and bandit feedback, provided that it has access to a suitable regret minimizer $\fR$.
	%
	% Both under full and partial feedback there is an algorithm which achieves regret $O(\poly(D)\sqrt{T})$.
	\Cref{alg:generalLinear_new} guarantees a cumulative regret $R_T \leq R_{T}^\fR(\co \vnu(\cX))$, where $R_{T}^\fR(\co \vnu(\cX))$ is the regret bound of a suitable regret minimizer $\fR$ for the set $\co \vnu(\cX)$.
\end{restatable}
In order to run \cref{alg:generalLinear_new}, one has to implement the $\textsc{Carath\'eodory}$ oracle and the inverse map $\vnu^\dagger$.
The following result shows that these admit efficient implementations in ``linear problems''.
These include as special cases many interesting settings, such as most online Bayesian persuasion problems studied in this paper.
\begin{restatable}{theorem}{LinearBanditoverPolytope}\label{th:LinearBanditoverPolytope}
	If $\cX$ is a polynomially-sized polytope and $\vnu$ is a linear map, i.e., there exists $\mM\in\R^{D\times M}$ such that $\nu(\xvec)=\mM\xvec$ for all $\xvec\in\cX$, then the $\textsc{Carath\'eodory}$ oracle and the inverse map $\vnu^\dagger$ can be implemented efficiently.
\end{restatable}
%
%
%\ma{Aggiungiere Cor nel caso in cui $\cX$ sia lineare e $\vnu$ sia un linear operator?}
%

Moreover, in the case of ``linear problems'' as in Theorem~\ref{th:LinearBanditoverPolytope}, we can instantiate Algorithm~\ref{alg:generalLinear_new} with specific regret minimizes $\mathfrak{R}$ for both the full and bandit feedback, so as to obtain the following guarantees.

\begin{restatable}{corollary}{FullFeedback}\label{cor:FullFeedback}
	Under the assumptions of Theorem~\ref{th:LinearBanditoverPolytope}, with full feedback, there exists a regret minimizer $\fR$ such that Algorithm~\ref{alg:generalLinear_new} is efficient and guarantees cumulative regret
	\[
	R_T\le \sqrt{DT}.
	\]
\end{restatable}

\begin{restatable}{corollary}{PartialFeedback}\label{cor:PartialFeedback}
	Under the assumptions of Theorem~\ref{th:LinearBanditoverPolytope}, with bandit feedback, there exists a regret minimizer $\mathfrak{R}$ such that Algorithm~\ref{alg:generalLinear_new} is efficient and guarantees cumulative regret
	\[
	R_T\le 16D^{3/2}\sqrt{T\log T}.
	\]
\end{restatable}

\section{Optimal Regret Bounds for Online Bayesian Persuasion with Partial Feedback}\label{sec:no reporting}

Next, we show that
%in Theorem~\ref{th:generalLinear} 
our general online learning framework introduced in Section~\ref{sec:smallLosses} can be applied to the setting of online Bayesian persuasion with partial feedback, enabling the derivation of novel state-of-the-art results.
%, which is the focus of this paper.
%
% of designing no-regret algorithms against a sequence of receivers selected by an oblivious adversary.

A standard revelation-principle-style argument shows that we can focus w.l.o.g. on signaling schemes that are {direct} and {persuasive} (see, \eg \citep{arieli2019private}).
%
% states that there always exists an optimal incentive-compatible signaling scheme which is \emph{direct} (see, \eg \citet{arieli2019private}).
%
In particular, a signaling scheme is \emph{direct} if signals correspond to action recommendations.
Formally, the set of signals of a receiver $r \in \rec$ is $\cS_r=A^{m}$, with each signal defining an action recommendation for each possible receiver $r$'s type.
%
% With an abuse of notation, in the following with denote receiver $r$'s signals as vectors $\avec\in A^m$. 
%
%Therefore, in order to compute an optimal incentive-compatible signaling scheme, we can focus on the sets of signals $\cS_r=A^{m}$ for each receiver $r$. In other words, a direct signal $s$ for receiver $r$ consists of an action recommendation for each of their possible $m$ types, that is a tuple $\avec\in A^m$. 
%
Moreover, a direct signaling scheme is \emph{persuasive} if each receiver's type is incentivized to follow the action recommendations issued by the sender. 
Formally, the set of direct and persuasive signaling schemes $\cP$ is the set of all $\phi :\Theta\to\Delta_{A^{m n}}$ such that, for every receiver $r\in\cR$, receiver $r$'s type $k\in\K_r$, and action $a \in A$, it holds
%
%We can define the set of direct \emph{persuasive} (\ie incentive-compatible) signaling schemes $\cP$ as the set of signaling schemes $\phi :\Theta\to\Delta_{A^{m^n}}$ such that, for each receiver $r\in\cR$, deviation $\hat a\in A$, and type $k\in\K_r$, it holds
\begin{equation}\label{eq:persuasive}
	\sum\limits_{\theta\in\Theta}\sum\limits_{\avec\in A^{m n}}\mu_\theta\phi^r_\theta(\avec)\mleft(u_k^{r}(a_k^r,\theta) - u_k^{r}( a,\theta) \mright)\ge 0,
\end{equation}
where, by slightly abusing notation, we denote as $A^{m n}$ the set $\cS$ with direct signals, while, given $\avec\in A^{m n}$, we let $a_k^r$ be the action in $\avec$ corresponding to type $k \in \K_r$ of receiver $r \in \rec$.
%
%and, for each $\theta\in\Theta$, \[\sum_{\avec \in A^{m^n}} \phi_{\theta}(\avec)=1.\]
%
% \matteo{define in 5.1 and 5.2?}
% %
% \begin{subequations}
% 	\begin{align}
% 	&\sum\limits_{\theta\in\Theta}\sum\limits_{\avec\in\A^{m}}\mu_\theta\phi_\theta(\avec)u_k^{r}(a^r_k,\theta)\ge\label{eq:IC_noTypeReporting} \\ 
% 	&\sum\limits_{\theta\in\Theta}\sum\limits_{\avec\in\A^{m}}\mu_\theta\phi_\theta(\avec)u_k^{r}(a^\prime,\theta),\, \forall r\in\mathcal{R}, a^\prime\in\A_r, k\in\K_r\nonumber\\
% 	&\sum_{\avec \in \A^m} \phi^k_{\theta}(\avec)=1,\, \forall \theta \in \Theta,\label{eq:Prob_noTypeReporting}
% 	\end{align}
% \end{subequations}
%
Intuitively, the inequality requires that, for a receiver $r$ of type $k$, the utility obtained by following recommendations given by $\phi$ is greater than or equal to that achieved by deviating to any another action $ a$.
%The second condition ensures that, for each state of nature, $\phi$ is a well-defined probability distribution over  action profiles.
%
%describes the persuasiveness constraints, which requires that for each receiver $r$ and type $k\in\K_r$, the value obtained by following the recommendations given by the signaling scheme $\phi$ is greater than or equal to the value for deviating to another action $\hat a$. The second condition ensures that, for each state of nature, $\phi$ is a well-defined probability distribution over  action profiles.
%
Notice that the set $\cP$ can be encoded as a polytope, by adding to the persuasiveness constraints those ensuring that $\phi$ is well defined, namely $\sum_{\avec \in \A^{m n}} \phi^k_{\theta}(\avec)=1$ for all  $\theta \in \Theta$.

Given any direct and persuasive signaling scheme $\phi\in\cP$, the sender's utility under type profile $\kvec\in\cK$ is 
\[
	\us(\phi,\kvec):=\sum_{\theta\in\Theta}\sum_{\avec\in A^{m n}}\mu_\theta\phi_\theta(\avec)\us((a_{k_1}^1,\ldots,a_{k_n}^n),\theta),
\]
where we remark that $a_{k_r}^r$ is the action recommendation specified by $\avec$ for a receiver $r$ whose realized type is $k_r$.
Moreover, let us observe that $\us(\phi,\kvec)$ is a linear function in the signaling scheme $\phi$.

As it is well known, finding an optimal direct and persuasive signaling schemes is $\mathsf{NP}$-hard, even when there is only one receiver and the distribution over receiver's types is known~\citep[Theorem 2]{castiglioni2020online}.
This implies that the polytope $\cP$ has exponential size, since the sender's utility can be represented as a linear function of direct and persuasive signaling schemes.
Moreover, classical reductions from offline to online optimization problems also show that there cannot be an efficient (\emph{i.e.}, with polynomial per-iteration running time) algorithm that achieves no-regret in this setting \citep{roughgarden2019minimizing,castiglioni2020online,daskalakis2022learning}.

A natural question is whether it is possible to design no-regret algorithm by relaxing the efficiency requirement on the per-iteration running time.
This question has already been answered affirmatively by~\citet{castiglioni2020online} in single-receiver settings.
In the following, we show that our online learning framework allows us to improve the regret bound in~\citep{castiglioni2020online} to optimality, by matching known lower bounds, and, additionally, it also allows us to extend the result to multi-receiver settings.

\subsection{Single-Receiver Setting under Partial Feedback}

Next, we consider the case of single receiver, \ie $n=1$.\footnote{In the single-receiver setting, we omit the dependence on $r$ from sets and other elements.}
In such a setting, the sender can observe a different loss for each of the $m$ different receiver's types.
Formally, the map $\vnu:\cP\to\R^{m}$ is defined by letting, for every $\phi\in\cP$:
\[
\vnu(\phi):=[-\us(\phi,k)]_{k \in \K}.
\]
% Now we can combine Theorem~\ref{th:generalLinear} and Theorem~\ref{th:LinearBanditoverPolytope} that guarantee that Algorithm~\ref{alg:generalLinear_new} combined with the regret minimizer $\mathfrak R$ for adversarial linear bandits of~\citet{abernethy2008competing}. Indeed we can instantiate it with the standard logarithmic barrier for the polytope $\vnu(\cP)$ which is a $m$-self concordant barrier~\citet[Proposition~5.4.1]{nesterov1994interior}.
% %
% The proof of the following theorem follows straightforward from the application of Theorem~\ref{th:generalLinear} and Theorem~\ref{alg:generalLinear_new}.
%
Then, we can apply \Cref{cor:PartialFeedback} to obtain the following regret upper bound under partial feedback. 
\begin{theorem}\label{th:Regret_singleReceiver}
	The single-receiver online Bayesian persuasion problem under partial feedback admits an algorithm which guarantees the following regret bound
	\[
	R_T=O(m^{2/3}\sqrt{T\log T}).
	\]
\end{theorem}
This result improves over the best known upper bound of $O(T^{4/5})$  by \citet[Theorem 4]{castiglioni2020online}.

\subsection{Multi-Receiver Setting under Partial Feedback}

\citet{castiglioni2021multi} introduce the online Bayesian persuasion problem with multiple receivers and adversarially-selected types.
They provide an algorithm that, under full feedback and some technical assumptions, guarantees sublinear regret.
In particular, their regret bound depends polynomially in the size of the problem instance when assuming that the number of possible receivers' type profiles is fixed.
This is a reasonable assumption given that the total number of type profiles is $|\K| = m^n$, which is exponential in the number of receivers $n$.
Under the same assumption, we provide the first no-regret algorithm under partial feedback.

Formally, we let $\overline\K \subseteq \K$ be the set of possible type profiles, so that, at each round $t \in [T]$, the receivers' type profile $\kvec_t$ belongs to $\overline\K$.
%
%the sender commits to a signaling scheme $\phi_t\in\cP$, and each receiver $r\in\rec$ is assigned a type $k_r\in\K_r$ such that $\kvec:=(k_1,\ldots,k_n)\in\overline\K$, 
%
%The maximum number of possible types is $m^n$, which is exponential in the number of receivers $n$.
%
% Here, as in the full-feedback setting of \citet{castiglioni2021multi}, 
%
We provide regret bounds which depend polynomially on the number of possible type profiles $|\overline \K|$. 
However, differently from~\citet{castiglioni2021multi}, in our algorithm working with partial feedback we assume that the set $\overline \K$ is known beforehand.
Indeed, an ``on the fly'' construction of $\overline \K$ as in~\citet{castiglioni2021multi} seems unfeasible under partial feedback, where, by definition, the sender does \emph{not} observe $\kvec_t$.

For every type profile $\kvec\in\overline\K$, the sender gets utility $\us(\phi, \kvec)$ by playing a signaling scheme $\phi$.
Then, we can define the map $\vnu:\cP\to \mathbb{R}^{|\overline{\K}|}$ so that, for every signaling scheme $\phi\in\cP$, it holds $\vnu(\phi):=[-\us(\phi,\kvec)]_{\kvec\in\overline\K}$.
Notice that $\vnu$ is a linear map from $\cP$ to $\R^{|\overline\K|}$. Thus, by \Cref{cor:PartialFeedback}, Algorithm~\ref{alg:generalLinear_new} gives the following regret bound.

\begin{theorem}\label{th:Regret_multiReceiver}
	The multi-receiver online Bayesian persuasion problem under partial feedback admits an algorithm which guarantees the following regret bound
	\[
	R_T=O\left(\left|\overline{\K}\right|^{2/3}\sqrt{T\log T}\right).
	\]
\end{theorem}

\section{Polynomial-Time Per-Iteration Running Time through Type Reporting}\label{sec:type reporting}
In this section, we show that it is possible to circumvent the negative results which rule out the existence of a no-regret algorithm for online Bayesian persuasion with polynomial per-iteration running time. We do that by enriching the decision space of the sender. 
In particular, we consider the framework of Bayesian persuasion \emph{with type reporting} introduced by~\citet{DBLP:conf/atal/CastiglioniM022} for offline settings, where the sender has the ability to commit to a \emph{menu} of signaling schemes, and then let the receivers choose their preferred signaling scheme depending on their private types.

\subsection{Online Type Reporting}

In the type-reporting model, at each round $t \in [T]$ of the repeated interaction, the sender proposes a \emph{menu} of marginal signaling schemes to each receiver.
We collectively denote them by $\varphi_t \defeq \{\varphi_t^{r,k}\}_{r \in \cR, k \in \cK_r}$, so that the menu proposed to receiver $r \in \cR$ consists of a set of marginal distributions $\varphi_t^{r,k} : \Theta \to \Delta_{\cS_r}$, one for each receiver's type $k \in \cK_r$. 
Then, each receiver $r \in \rec$ reports a type $k_r \in \cK_r$ to the sender. The reported type $k_r$ is such that the signaling scheme $\varphi_t^{r,k_r}$ is the one 
guaranteeing to the receiver the highest expected utility among those in the menu.\footnote{Such step can be equivalently implemented by extending the interaction between the sender and the receiver: the sender can ask each receiver $r \in \cR$ to directly select a marginal signaling scheme $\varphi_t^{r,k}$ from the menu, and the receiver will be incentivised to select the one corresponding to its own type $k_r$.}
Finally, the sender computes and commits to the signaling scheme $\phi_t: \Theta \to \Delta_\cS$ which maximizes the sender's expected utility among the signaling schemes whose marginals are equal to the marginal signaling schemes $\varphi_t^{r,k_r}$ corresponding to the types $k_r$ reported by the receivers,
%
%selected by each receiver $r\in\cR$ of type $k_r$, 
\ie $\phi_t^r=\varphi^{r,k_r}$ for every $r\in\rec$.
From this point on, the interaction goes on as in the case without type reporting.

Notice that, in the type-reporting setting, the sender observes the types of the receivers at each round $t \in [T]$. Thus, in the type-reporting model, the sender always has full feedback.

Let us also remark that the assumption that the sender can only propose marginal signaling schemes to the receivers is w.l.o.g., since the expected utility of each receiver only depends on their marginal signaling scheme, and \emph{not} on those of the others (see \Cref{sec:preliminaries}).
Therefore, the sender can delay the choice of the joint signaling scheme $\phi_t$ until after all the receivers reported their types.

%Therefore, in a type-reporting setting, the set of sender's strategies encompasses all the possible singling schemes $\phi$ that marginalized to the menus of the reported types $\vk\in\cK$.
%%\ie $\varphi = \{\varphi^{r,k}\}_{r \in \cR, k \in \cK_r}$.
%%
%Thus, given a menu $\varphi = \{\varphi^{r,k}\}_{r \in \cR, k \in \cK_r}$ and a type's tuple $\vk$, the resulting sender's expected utility is $\us (\phi, \vk)$,
%%Given one such strategy, by overloading notation, whenever the receivers report a type profile $\vk \in \cK$, we denote the resulting sender's expected utility as follows:
%%\begin{equation}
%%\us (\phi, \vk) \coloneqq \max_{\phi : \phi^r = \varphi^{r,k_r}} \us(\phi, \vk),
%%\end{equation}
%% where the maximization accounts for the fact that, once marginal schemes are fixed, the sender only needs to optimize their utility among the (overall) signaling schemes which are consistent with marginal ones.
%where the sender selects a ``joint'' signaling scheme $\phi$ which has marginal $\phi^r$ consistent with $\varphi^{r,k_r}$ for each receiver $r$, \ie $\phi^r=\varphi^{r,k_r}$ for all $r\in\rec$.

%
By a revelation-principle-style argument~\citep{castiglionibayesian}, it is always possible to focus w.l.o.g. on \emph{incentive compatible} (IC) menus $\varphi=\{\varphi^{r,k}\}_{r\in\cR,k\in\cK_r}$, which are those such that each receiver $r \in \cR$ is incentivized to report their true type, say $k_r \in \cK_r$.
Formally, for all  $k \neq k_r \in \cK_r$,
\begin{equation}
	\sum_{s_r \in \cS_r}  \max_{a \in \cA_r} \sum_{\theta \in \Theta} \mu_\theta \, \varphi_\theta^{r,k_r}(s_r) \, \ur[r][k](a,\theta) \geq
	\sum_{s_r \in \cS_r}  \max_{a \in \cA_r} \sum_{\theta \in \Theta} \mu_\theta \, \varphi_\theta^{r,k}(s_r) \, \ur[r][k](a,\theta),\label{eq:ICandPers_Menu}
\end{equation}
where the $\max$ operators account for the fact that the receiver plays a best-response action after receiving a signal.

W.l.o.g., we can focus on menus that are \emph{direct}, namely $\cS_r = A$ for every $r \in \rec$, and \emph{persuasive}.
We say that a direct menu $\varphi=\{\varphi^{r,k}\}_{r\in\cR,k\in\cK_r}$ is {persuasive} if the marginal signaling schemes $\varphi^{r,k}$ satisfy persuasiveness constraints similar to those of Equation~\eqref{eq:persuasive} for every receiver $r \in \rec$ and type $k\in\cK_r$.
Then, we define $\Lambda$ as the set of menus which are IC,
%(\ie \Cref{eq:ICandPers_Menu} holds), 
direct, and persuasive.
% for a receiver $r \in \rec$.
%
%In order to compute an optimal menu in $\Lambda_r$, it is enough to focus on menus which only comprise direct signaling schemes, that is, we can consider $\cS_r =A$.

%\footnote{Since we have only one receiver we can drop the index $r\in\rec$ for ease of notation.}

The sender's goal is to compute a sequence of IC menus $\left\{ \varphi_t \right\}_{t \in [T]}$ and a sequence of signaling schemes $\left\{ \phi_t \right\}_{t \in [T]}$ which are consistent with the menus, whose performance over the $T$ rounds is measured in terms of the following notion of regret:
\[
R_T \coloneqq \max_{\varphi} \sum_{t = 1}^T\us(\varphi, \vk_t) - \sum_{t=1}^T \mathbb{E}\left[\us(\phi_t, \vk_t)\right],
\]
where, by overloading notation, we denoted with
\begin{equation}\label{eq:SenderUtilityOnMenu}
\us (\varphi, \vk) \coloneqq \max_{\phi : \phi^r = \varphi^{r,k_r}} \us(\phi, \vk)
\end{equation} 
the maximum utility of the sender when the receivers' type profile is $\vk \in \cK$.
We remark that the above formulation of regret is stronger than the classical one in which a best-in-hindsight decision is fixed for all the rounds.
Indeed, although the best menu $\varphi$ is fixed for all $t\in[T]$, we allow the signaling scheme $\phi^\star_t\in\arg\max_{\phi : \phi^r = \varphi^{r,k_{t,r}}} \us(\phi, \vk_t),$ to depend on the round $t$, as long as $\phi^\star_t$ has fixed marginals that are compatible with the best menu $\varphi$.
%!
\subsection{Single-Receiver Setting with Type Reporting}\label{sec:singleTypeRep}

We start by studying the single-receiver setting (\ie $n=1$). 
%
%We say that a menu $\varphi=\{\varphi^{r,k}\}_{r\in\cR,k\in\cK_r}$ is persuasive if, for each receiver $r$ and type $k\in\cK_r$, the marginal signaling scheme $\varphi^{r,k}$ is persuasive for $r$ under type $k$. 
%%
%Then, we define $\Lambda$ as the set of menus which are IC (\ie \Cref{eq:ICandPers_Menu} holds) and persuasive. In order to compute an optimal menu in $\Lambda$, it is enough to focus on menus which only comprise direct signaling schemes, that is, we can consider $\cS=A$.\footnote{Since we have only one receiver we can drop the index $r\in\R$ for ease of notation.}
%
%In such a setting, the value of a menu $\varphi$ for the sender is simply $\us(\varphi,k)$ as the overall signaling scheme $\phi$ is precisely the signaling scheme $\varphi^k$.

In the type-reporting setting it is not possible to directly write a succinct representation of the set of persuasive menus to obtain a polytope with polynomial size, as it was the case in previous sections. The reason for this is that encoding the inner maximizations of Equation~\eqref{eq:ICandPers_Menu} as a set of linear inequalities would require exponentially-many constraints.
However, this observation does \emph{not} rule out the existence of efficient algorithms. Indeed, even if $\Lambda$ has an exponential description, it is possible to show that it has polynomial \emph{extension complexity}~\citep{fiorini2012extended}. In particular, we can show that there exists a succinct representation of $\Lambda$ in a suitable higher dimensional space. This was already implicitly shown by \citet{castiglionibayesian}, here we provide a formal characterization for completeness.

Intuitively, the construction works as follows: we introduce extra variables $l$, called extension variables such that the extended polytope $\cL$ is defined by variables $\ell\equiv(\varphi, l)$, where $\varphi_\theta^{k}\in\mathbb{R}_+^{|\mathcal{A}|}$ for each $\theta\in\Theta,k\in\K$, and we have one variale $l_a^{k,k^\prime}\in\R$ for each $a\in A,k, k^\prime\in\K$. The polytope $\cL$ can be described by a polynomial numer of constraints. This fact, together with the linear projection map $\pi:\cL\to\Lambda$ defined as $\pi(\varphi,l)=\varphi$, proves the polynomial extension complexity of $\Lambda$.
Formally, the extended polytope $\cL$ can be described by the following inequalities:
%However we can consider an extended formulation $\mathcal L$ of such polytope that has polynomial size. We achieve this by introducing extra variables $l_a^{k,k^\prime},\,\forall a\in A,k=k^\prime\in\K$.
%
%The extended polytope $\cL$ is the set of variables $\ell\equiv(\varphi, l)$ where $\varphi_\theta^{k}\in\mathbb{R}_+^{|\mathcal{A}|},\,\forall \theta\in\Theta,k\in\K$ and $l_a^{k,k^\prime},\,\forall a\in A,k, k^\prime\in\K$ that satisfies the following linear inequalities.
%
\begin{subequations}\label{lambdaTRsingle}
	\begin{align}
	&\sum_{\theta \in \Theta} \sum_{a \in A} \mu_\theta \varphi_\theta^{k}(a) \ur[][k](a,\theta) \geq\sum\limits_{a\in A}l_a^{k,k^\prime},\forall{k,k' \in \K_r} \label{eq:lambdaTRsingle1}\\
	& l_a^{k,k^\prime}\hspace{-0.1cm}\ge\hspace{-0.1cm} \sum_{\theta \in \Theta} \mu_\theta  \varphi_\theta^{k'}(a)   \label{eq:lambdaTRsingle2} \ur[][k](a^\prime,\theta),\forall k,k'  \in \cK_r, a,a' \hspace{-0.1cm}\in  A\\
	&\sum_{a \in A} \varphi^k_{\theta}(a)=1,\, \forall k \in \K_r,\forall \theta \in \Theta, \label{eq:lambdaTRsingle3}
	\end{align}
\end{subequations}
where $l^{k,k'}_{a}$ represents the maximum utility received by a receiver of type $k$ but reporting type $k'$, when type $k'$ is recommended action $a$. 

Then, we instantiate Algorithm~\ref{alg:generalLinear_new} by taking the set $\cL$ as the polytope $\cX$, where we have one loss for each of the $m$ types that can be reported by the receiver. 
We define $\vnu:\cL\to\R^m$ as the vector valued map mapping each feasible point $\ell=(\varphi,l)$ into the $m$-dimensional vector of losses $\vnu(\ell):=[-\us (\varphi^{k}, k)]_{k \in \K}$, where the value of a menu $\varphi$ for the sender against a receiver's type $k \in \K$ is $\us(\varphi^k,k)$ as the overall signaling scheme $\phi$ coincide with the signaling scheme $\varphi^k$, when $n=1$.
Then, \Cref{cor:FullFeedback} yields the following result. 

\begin{restatable}{theorem}{typereportingsingle}\label{thm: type rep single}
	The single-receiver online Bayesian persuasion problem with type reporting admits an algorithm which guarantees regret $R_T\le\sqrt{mT}$ and polynomial per-iteration running time. 
\end{restatable}
\subsection{Multi-Receiver Setting with Type Reporting}\label{sec:multiTypeRep}
In this section, we focus on the problem of designing a no-regret algorithm for the multi-receiver setting with type reporting.
The method employed in the case of a single receiver is not applicable here, as the number of possible type profiles becomes exponentially large, resulting in exponentially many possible loss functions.
Moreover, it is not possible to directly design efficient algorithms working on the joint action space since it has exponential size.
In order to build a no-regret algorithm for this setting, the idea is to cast the learning problem into a decision space which is small enough to be manageable. In particular, we observe that the sender must commit only to the marginal signaling schemes $\{\varphi^{r,k_r}_t\}_{r \in \rec,k_r \in \K_r}$ before observing the receivers' types. Then, at each round $t$, the sender receives the types $k_{r,t}$ for each receiver $r\in\cR$, and solves an offline optimization problem to compute the optimal joint signaling schemes $\phi_t$ that has marginal signaling schemes $\{\varphi^{r,k_{r,t}}_t\}_{r \in \rec}$. %the receivers' types in $\kvec_{t}$.
By exploiting this observation, we develop a no-regret algorithm that operates within the smaller decision space of marginal signaling schemes.

%Similarly to Section~\ref{sec:singleTypeRep}, 
Let $\Lambda_r$ be the set of IC and direct menus of marginal signaling schemes for receiver $r \in \rec$. Formally, $\Lambda_r$ is defined as the set of $\varphi^{r,k}$ satisfying Equations~\eqref{eq:lambdaTRsingle1} -- \eqref{eq:lambdaTRsingle3} for every receiver $r \in \rec$ and type $k \in \K_r$. Moreover, let $\Lambda\coloneqq  \bigtimes_{r \in \rec} \Lambda_r$. Intuitively, an element of $\Lambda$ includes a menu of marginal signaling schemes $\varphi^r$ for each receiver $r \in \rec$.
Then, the action space of the learner is given by the set of IC and persuasive marginal signaling schemes $\Lambda$. The sender's utility when the agents are of type $\kvec \in \K$ is defined by a function $g^{\kvec}:\Lambda\rightarrow [0,1]$, where $g^\kvec(\varphi)$ is the value obtained by the following linear program which is an expansion of the maximization in Equation~\eqref{eq:SenderUtilityOnMenu}:
\begin{subequations}\label{lp:multi}
	\begin{align}
		&\max_{\phi \ge 0}  \,\, \sum_{\theta \in \Theta}\sum_{\va \in \A} \mu_\theta  \phi_\theta(\va) \us_\theta(\va) \quad \text{s.t.} \label{obj:lpMulti} \\
		& \sum_{\substack{\va \in \A:\\ a_i= \hat a}} \phi_\theta(\va)  = \varphi^{r,k_r}_\theta(\hat a),\, \forall r \in \rec, \hat a \in \A_r, \theta \in \Theta. \label{lp:multi1} 
	\end{align}
\end{subequations}
where Equation~\eqref{obj:lpMulti} is the utility of a signaling scheme $\phi$ and Equation~\eqref{lp:multi1} encodes the constraints on the signaling scheme $\phi$ to have marginals $\{\varphi^{r,k_{r}}\}_{r \in \rec}$.
%
%A property that is often used to establish good regret bounds is the concavity of the reward function.
%
The function $g^\kvec(\varphi)$ is the solution to a parametric (in $\varphi$) linear program. If we want to solve an online problem involving $g^\kvec$, we first have show that the offline problem $\max_{\varphi\in\Lambda} g^\kvec(\varphi)$ is in some sense computationally tractable. More precisely, we show that for any $\kvec\in\K$ the function $g^{\kvec}$ is concave.

%\matteo{In the following, given two $\varphi, \hat \varphi\in \Lambda$ we denote with $\bar \varphi=\varphi+\hat\varphi$ the componentwise sum in which  $\bar \varphi^{r,k_r}_\theta(a_1)=\varphi^{r,k_r}_\theta(a_1)+\hat \varphi^{r,k_r}_\theta(a_1)$. Moreover, given a scalar $q$ and a $\varphi\in \Lambda$, $q\varphi$ denotes the vector in which all the components are multiplied by $q$. Similar operations hold for $\phi$ }
%\matteo{NON SO SE E' NECESSARIO.}

\begin{restatable}{lemma}{lemmaConcave} \label{lem:concave}
	The function $g^{\kvec}(\varphi)$ is concave in $\varphi$ on $\Lambda$ for each type profile $\kvec \in \K$.
\end{restatable}
Moreover, we show that the function is particularly well behaved. In particular, we prove that it is Lipschitz-continuous with respect to the $\ell_2$ norm. This will be useful to upperbound the norm of gradients of the function $g^{\kvec}$.

\begin{restatable}{lemma}{lemmaCont}\label{lm:cont}
	For each $\kvec \in \K$, the function $g^{\kvec}(\varphi)$ is $\sqrt{nd|A|}$-Lipschitz-continuous in $\varphi$ with respect to $\| \cdot\|_2$.
\end{restatable}

Since we have no access to the gradient of the functions $g^{\kvec}$, a natural choice to implement a no-regret algorithm is to apply Follow the Regularized Leader (FTRL) \citep{abernethy2008competing,hazan2010extracting}.
Algorithm~\ref{alg:typeReporting} describes the specific implementation of the FTRL-type algorithm.
At each iteration the algorithm proposes a set of IC menus of marginal signaling schemes $\varphi_t\in \Lambda$.
Then, the algorithm observes the reported types $\kvec_t$ (notice that the receivers report their true types since the menu is IC).
The algorithm computes a signaling scheme $\phi$ solving LP~\eqref{lp:multi} for the types  $\kvec_t$, returning a signaling scheme with value $g^{\kvec}(\varphi_t)$.
Finally, the algorithm updates the set of menus of signaling schemes by computing:
\begin{equation}\label{eq:FTRL}
\varphi_{t+1}= \arg \max_{\varphi \in \Lambda}\sum_{\tau \in [t]} g^{\kvec_\tau}(\varphi)-\frac{1}{2\alpha} \lVert\varphi\rVert_2^2.
\end{equation}

\begin{algorithm}[tb]\caption{\textsc{no-regret algorithm type-reporting}}\label{alg:typeReporting}
	\begin{algorithmic}[1]
		\REQUIRE  any set of marginal signaling schemes $\varphi_1 \in \Lambda$, learning rate $\alpha$ 
		\FOR{ $t=1$ to $T$}
		\STATE propose the set of menus of signaling schemes $\varphi_t$
		\STATE observes the receivers reported types $\kvec_t$
		\STATE {$\phi_t\gets $ a solution of LP~\eqref{lp:multi} for $\varphi_t$  with value  $g^{\kvec_t}(\varphi_t)$ \label{line4} }
		\STATE $\varphi_{t+1}\gets\arg \max_{\varphi \in \Lambda} \sum_{\tau\le t} g^{\kvec_\tau}(\varphi)- \frac{1}{2\alpha}\lVert\varphi\rVert_2^2  $\label{line5}
		\ENDFOR
	\end{algorithmic}
\end{algorithm}

Following the standard FTRL analysis we can provide an upper bound on the regret for Algorithm~\ref{alg:typeReporting}. 

\begin{restatable}{theorem}{regretMultiType}
	Let $\alpha\defeq \sqrt{{m}/{T}}$. Algorithm~\ref{alg:typeReporting} guarantees a cumulative regret
	$
	R_T\le nd|A|\sqrt{mT}.
	$
\end{restatable}

\subsection{An efficient Implementation for Multi-Receiver Online Bayesian Persuasion with Type Reporting}\label{sec:multiTypeRepEfficient}

In the previous section, we provided a no-regret algorithm for the multi-receiver problem. However, we did not address the question of whether \Cref{alg:typeReporting} can be implemented efficiently. Specifically, determining $\phi_t$ and $\varphi_{t+1}$ (Line 4 and 5, respectively) is not straightforward.
In general, the sender's utility function cannot be represented in space polynomial in the number of players. For this reason, computational works on multi-receiver Bayesian persuasion focus on succinctly representable utility functions (see, \eg \citep{dughmi2017algorithmic,babichenko2017algorithmic,castiglioni2021multi}).
In particular, each receiver's action set $A$ is binary, and the two actions are denoted by $a_1$ and $a_0$. Then, the sender's utility function can be compactly represented as $\fs_\theta(R)$, where $R\subseteq \rec$ is the set of receivers playing $a_1$.
The literature we just mentioned examines three common types of utility functions: \emph{supermodular}, \emph{submodular}, and \emph{anonymous}.
For the case of submodular functions, it is well known that even in the offline setting without types, the problem is \NPHard~to approximate up to within any factor better than $(1-\sfrac{1}{e})$~\citep{babichenko2017algorithmic}.
Therefore, in this section, we show that Algorithm~\ref{alg:typeReporting} can be implemented efficiently when the sender's utility function is monotone, supermodular, or monotone, anonymous.

\begin{definition}
% We say that the function $\fs_\theta$ is \emph{submodular} if, for $R,R' \subseteq \rec $, 
% \[
% \fs_\theta(R\cap R')+\fs_\theta(R\cup R')\le \fs_\theta(R)+\fs_\theta(R').
% \]
The function $\fs_\theta$  is \emph{supermodular} if, for $R,R' \subseteq \rec $,
\[
\fs_\theta(R\cap R')+\fs_\theta(R\cup R')\ge  \fs_\theta(R)+\fs_\theta(R').
\] 
Finally, the function $\fs_\theta$ is \emph{anonymous} if $\fs_\theta(R) = \fs_\theta(R')$ for all $R,R' \subseteq \rec$ such that $|R| = |R'|$.
\end{definition}

We show that we can efficiently solve LP~\eqref{lp:multi} and the concave program of Equation~\eqref{eq:FTRL} (which both have an exponential number of variables, but polynomially many constraints) by writing their dual formulation, and then using the ellipsoid method with a suitable efficient separation oracle.

As a separation oracle, we use the following general optimization oracle.
\begin{definition}[Optimization Oracle]\label{def:opt}
	Given in input a function $\fs$ and a vector of weights $w\in\R^n$, with $w_r$ denoting the component corresponding to receiver $r$, an \emph{optimization oracle} $\cO$ returns a subset of receivers such that 
	\[
		\cO(\fs_\theta,w)\in\arg \max_{R \subseteq \rec} \left\{\fs_\theta(R)+\sum_{r\in R} w_r\right\}.
	\]
	% An optimization oracle, given in input a function $\fs$ and a set of weights $w \in \mathbb{R}^n$ with $w_r$ denoting the component corresponding to receiver $r$, return a subset $R\subset \rec$ in
	% \[
	% \cO(\fs_\theta,w)\coloneqq\arg \max_{R \subseteq \rec} \left\{\fs_\theta(R)+\sum_{r\in R} w_r\right\}.
	% \]
\end{definition}
%
%It is well known~\citep{dughmi2009submodular, dughmi2017algorithmic} that $\cO(\fs_\theta,w)$ can be implemented in polynomial-time both when $\fs_\theta$ is supermodular or anonymous utility functions.
%
Moreover, will will use the following known result.

\begin{lemma}[\citet{babichenko2017algorithmic} and \citet{dughmi2017algorithmicExternalities}]\label{lem:dughmi}
	The optimization oracle $\cO(\fs_\theta,w)$ can be implemented in polynomial-time when $\fs_\theta$ is a supermodular or anonymous monotone utility function.
\end{lemma}

In the following, we show that when we have access to the separation oracle $\cO$, both the optimization problem in Line~4 and Line~5 can be solved in polynomial-time using the ellipsoid method. We start by providing a polynomial-time algorithm for LP~\eqref{lp:multi}.
Intuitively, the problem is equivalent to that of finding an optimal signaling scheme in a problem with fixed marginal signaling schemes. In particular, by rewriting LP~\eqref{lp:multi} for the specific case of a binary action space and by taking its dual, we obtain
	\begin{align*}
		\min_{x} &  \sum_{r \in \rec, \theta \in \Theta} \varphi^{r,k_r}_\theta(a_1) x_{r,\theta}\quad \textnormal{s.t.}\\
		&\sum_{r\in R} x_{r,\theta} \ge \mu_\theta \fs_{\theta}(R),\quad \forall R\subseteq \rec,\theta \in \Theta,
	\end{align*}
where the dual variables are $\{x_{r,\theta}\}_{r \in \rec,\theta \in \Theta}$ (more details on the derivation are provided in \Cref{app:MultiTypeRepEfficient}).
A separation oracle for the dual problem above can be implemented applying the optimization oracle $\cO(\fs_\theta, -x_{\theta}/\mu_\theta)$ for each state of nature $\theta\in\Theta$. Let $R^\ast_\theta\defeq \cO(\fs_\theta, -x_{\theta}/\mu_\theta)$. If there exists $\theta$ such that
\[\fs_\theta(R^\ast_\theta)-\sum_{r\in R^\ast_\theta}\frac{x_{r,\theta}}{\mu_\theta}\ge0,\]
then we can use the violated constraint $(\theta,R^\ast_\theta)$ as a separating hyperplane.
Then, we can run the ellipsoid method equipped with such separation oracle on the dual of LP~\eqref{lp:multi}. This procedure, together with known properties of the ellipsoid method (see, \eg \citep{khachiyan1980polynomial,grotschel2012geometric}), yields the following result. 	

\begin{restatable}{lemma}{ellipsoidEasy} \label{lm:oracle1}
	Given access to an optimization oracle $\mathcal{O}$, there exists a polynomial-algorithm that solves LP~\eqref{lp:multi}.
\end{restatable}

Next, we prove that the concave program of \Cref{eq:FTRL} can be solved efficiently when having access to the optimization oracle $\cO$.
In order to solve the concave program of \Cref{eq:FTRL}, we start by rewriting the problem on the space of joint signaling schemes $\phi$. To do that, we need to introduce constraints that ensure that the joint signaling scheme $\phi$ is well-defined with respect to marginals $\varphi$ (see \Cref{lp:multiBayes} in \Cref{app:MultiTypeRepEfficient}).
Then, we compute the Lagrangian relaxation of the resulting problem. By noticing that the problem is concave, and that Slater's condition holds, we recover strong duality. 
Finally, we use KKT conditions to remove the exponentially-many variables $\phi$, and thereby obtaining a concave optimization problem with polynomially-many variables and exponentially-many constraints. Applying a similar procedure to the one we used for \Cref{lm:oracle1}, we can solve such problem via the ellipsoid algorithm by using the oracle $\cO$ of Definition~\ref{def:opt} as a separation oracle.
%, Similarly to Lemma~\ref{lm:oracle1} we can prove the following statement:

\begin{restatable}{lemma}{quadratic}\label{lem:oraclequadratic}
Given access to an optimization oracle $\mathcal{O}$, there exists a polynomial-time algorithm that solves the problem of \Cref{eq:FTRL}.
\end{restatable}

% As a consequence of the previous two lemmas, we obtain the following theorem.
As a consequence of Lemma~\ref{lem:dughmi}, Lemma~\ref{lm:oracle1} and Lemma~\ref{lem:oraclequadratic} we can conclude the following:

\begin{restatable}{theorem}{corollaryOracle}
	In settings in which receivers have binary actions, and the sender has a monotone, supermodular or a monotone, anonymous utility function, \Cref{alg:typeReporting} has polynomial per-iteration running time and guarantees
	\[
	R_T\le nd|A|\sqrt{mT}.
	\]
\end{restatable}

\section{Further Applications}

The main motivation for introducing the reduction from online problems with finite number of losses to online linear optimization of Section~\ref{sec:smallLosses} was to solve online Bayesian persuasion problems. In this section, we highlight two further applications of our framework beyond Bayesian persuasion.

\paragraph{Online Learning in Security Games}
\citet{balcan2015commitment} extended classic (one-shot) security games (see, \eg \citet{tambe2011security}) by introducing the problem of learning a no-regret strategy for the defender against an adversarial sequence of attackers. In their model, at each round $t$, the defender chooses a strategy $\xvec_t$, which is a distribution over $N$ targets. Then, an attacker of type $d_t\in D$, best responds to such strategy and the defenders experience a loss of $L_{d_t}(\xvec_t)$. Our reduction yields a $\tilde O(\poly(D)\sqrt{T})$ regret bound under partial feedback, which improves the regret bound given in~\citet{balcan2015commitment}, which is of order $ O(\poly(ND)T^{2/3})$.

\paragraph{Online Bidding in Combinatorial Auction}
\citet{daskalakis2022learning} studied online learning in repeated combinatorial auctions.
In these auctions the action space is combinatorial and, therefore, exponentially large. However, \citet{daskalakis2022learning} show that whenever the different number of bid profiles of the other bidders is finite and small (of size $D$), it is possible to design $O(\sqrt{DT})$ regret algorithms under \emph{full feedback}. Our reduction to online linear optimization allows us to match their bound with full-information feedback, and also gives a $\tilde O(\poly(D)\sqrt{T})$ bound for the more realistic case of \emph{partial feedback}, \ie each player only observes their own utility.
%
%
%
%They first reduce the problem of achieving no-regret to the problem ‘‘buyer’s problem'' in which a buyer has to pick a subset of elements to buy, while an adversary picks prices for them. In the case these possible prices are from a finite set of size $D$ they provide a no-regret algorithm for the online bidding problem which achieves $O(\sqrt{DT})$ regret bound, assuming the existence of an offline optimization oracle and full-information. Our reduction to online linear optimization it's able to match their bound with full-information and also gives the same bound of $O(\poly(D)\sqrt{T})$ for the partial information setting, \ie in which the player only observes their utility.

\clearpage

\bibliographystyle{plainnat}
\bibliography{biblio}

\clearpage

\newpage
\appendix
\onecolumn
\section{Proofs Omitted from Section~\ref{sec:smallLosses}}

\generalLinear*
\begin{proof}
	First, notice that, given any $\zvec_t\in\co \vnu (\cX)$, thanks to Carathéodory's theorem there always exist $D+1$ points $\{\zvec^1_t,\ldots, \zvec^{D+1}_t\}\subset \vnu(\cX)$ and a corresponding probability distribution $\vlambda = (\lambda^1_t,\ldots, \lambda^{D+1}_t)\in\Delta^{D+1}$ such that ${\zvec}_t=\sum_{i=1}^{D+1}\lambda^i_t \, \zvec_t^i$.
	Such points $\zvec_t^i$ with their corresponding probabilities $\lambda_t^i$ are those returned by the procedure $\textsc{Carath\'eodory}(\zvec_t, \vnu(\cX))$ called by Algorithm~\ref{alg:generalLinear_new}.
	Thus, given how the algorithm selects the $\vx_t\in \cX$ to be played at each $t \in [T]$, it holds $\E \left[ L_{d_t}(\vx_t) \right] = \vnu(\vx_t)^\top \onevec_{d_t} =  \zvec_t^\top \onevec_{d_t}$.
	
	Second, by using the no-regret property of the regret minimizer $\fR$, the following holds:
	%	\begin{align*}
		%		R_T & = \sum_{t=1}^T\Exp [L_{d_t}(\vx_t)]-\min_{\vx \in\cX}\sum_{t=1}^T L_{d_t}(\vx)\\
		%		&=\sum_{t=1}^T \zvec_t^\top\onevec_{d^t}-\min_{\zvec\in\vnu (\cX)}\sum_{t=1}^T\zvec^\top\onevec_{d^t}\\
		%		&\le\sum_{t=1}^T \zvec_t^\top\onevec_{d^t}-\min_{\zvec\in\co \vnu(\cX)}\sum_{t=1}^T\zvec^\top\onevec_{d^t}\\
		%		&=O(\poly(D)\sqrt{T}),
		%	\end{align*}
	\begin{align*}
		R_T & = \sum_{t=1}^T\Exp [L_{d_t}(\vx_t)]-\min_{\vx \in\cX}\sum_{t=1}^T L_{d_t}(\vx)\\
		&=\sum_{t=1}^T \zvec_t^\top\onevec_{d_t}-\min_{\zvec\in\vnu (\cX)}\sum_{t=1}^T\zvec^\top\onevec_{d_t}\\
		&\le\sum_{t=1}^T \zvec_t^\top\onevec_{d_t}-\min_{\zvec\in\co \vnu(\cX)}\sum_{t=1}^T\zvec^\top\onevec_{d_t}\\
		&\le R_T^{\fR}(\co\vnu(\cX))
	\end{align*}
	where the first inequality holds since $\vnu(\cX) \subseteq \co \vnu(\cX)$.
\end{proof}

\LinearBanditoverPolytope*

\begin{proof}
	If $\cX$ is a polytope and $\vnu$ is a linear map then $\vnu(\cX)$ is a polytope, and thus elements of $\co \vnu(\cX)$ correspond to elements of $ \vnu(\cX)$. Therefore, the $\textsc{Carath\'eodory}$ oracle can be implemented as just returning the one point density at $\zvec$ for every $\zvec\in\co \vnu(\cX)$. 
	
	Moreover, since $\vnu$ is linear we can implement $\vnu^\dagger$ by computing a generalized inverse of its matrix representation $\mM$, and produce $\vnu^\dagger(\zvec)=\mM^\dagger \zvec\in\cX$. By definition of generalized inverse that holds for all $\zvec\in\vnu(\cX)$, \ie there exists an $\xvec$ such that $\mM\xvec=\zvec$, we have that 
	\[
	\vnu(\vnu^\dagger(\zvec))=\mM\mM^\dagger\zvec=\mM\mM^\dagger \mM\xvec=\mM\xvec=\zvec,
	\]
	which concludes the proof.
\end{proof}

\FullFeedback*
\begin{proof}
	We can set $\fR$ to be Online Gradient Descent (OGD)~\citep{zinkevich2003online}. Indeed, we have that the gradient of the losses in $\co\vnu(\cX)$ is bounded by 1 in the $\ell_2$-norm, and that $\co\vnu(\cX)\subset[0,1]^D$, which gives a $D$ bound on the diameter w.r.t. the the $\ell_2$-norm. Thus, by setting the learning rate of OGD as $\sqrt{D/T}$ we obtain a regret bound of $R_T^{\fR}(\co \vnu(\cX))\le \sqrt{DT}$~\citep{Orabona}.
\end{proof}

\PartialFeedback*
\begin{proof}
	Under partial feedback, we obtain the regret bound above by equipping Algorithm~\ref{alg:generalLinear_new} with a suitably-defined regret minimizer $\fR$.
	In particular, $\fR$ must work by observing only realizations of an unbiased estimator of $\zvec_t^\top \onevec_{d_t}$ instead of its actual value, since Algorithm~\ref{alg:generalLinear_new} does \emph{not} play $\zvec_t$, but it employs a sampling process that is equivalent to playing $\zvec_t$ in expectation.
	Such a regret minimizer $\fR$ can be implemented by the algorithm introduced by~\citet{abernethy2008competing}, as any polytope in $\R^D$ has a $D$-self concordant barrier~\citet[Theorem~2.5.1]{nesterov1994interior}. This yields
	$
	R_T^{\fR}(\co\vnu(\cX))\le16D^{3/2}(T\log T)^{1/2}
	$, which proves our statement.
\end{proof}

\section{Proofs Omitted from Section~\ref{sec:singleTypeRep}}

\typereportingsingle*
\begin{proof}
	By \Cref{cor:FullFeedback}, \Cref{alg:generalLinear_new} produces a sequence $(\ell_t)_{t=1}^T,\ell_t\in\cL$, such that
	\[
	\sum\limits_{t=1}^T \vnu(\ell_t)^\top\onevec_{k_t}-\min\limits_{\ell\in\cL}\sum\limits_{t=1}^T\vnu(\ell)^\top\onevec_{k_t}\le \sqrt{mT}.
	\]
	Then, the sender commits to the menu which is the projection of $\ell_t$ onto $\Lambda$, \ie $\varphi_t=\pi(\ell_t)$. Since $\vnu(\ell)$ is independent from the extension variables $l$ we get that:
	\[
	\vnu(\ell_t)^\top\onevec_{k_t}=-\us(\pi(\ell_t), k_t)=-\us(\varphi_t, k_t)
	\]
	and similarly
	$
	\vnu(\ell)^\top\onevec_{k_t}=-\us(\varphi, k_t),
	$
	which proves the statement. 
\end{proof}
\section{Proofs Omitted from Section~\ref{sec:multiTypeRep}}\label{app:MultiTypeRep}

\begin{lemma}\label{lem:bersimas}
	For any $\varphi\in\Lambda$ we can write $g^\kvec(\varphi)$ as a solution of a standard-form linear program with $|\A|\cdot|\Theta|$ variables and constraints, and in such a standard-form linear program, the variables $\varphi$, are its \emph{right-hand} side vector.
\end{lemma}

\begin{proof}
	We define a standard form linear program with $n$ variables and $n$ constraints if it is of the form:
	\begin{align*}
		&\max\limits_{\xvec}\vc^\top\xvec,\,s.t.\\
		&\mA\xvec=\vb, \xvec\ge 0,
	\end{align*}
	where $\xvec,\vb,\vc\in\R^n$ and $\mA\in\R^{n\times n}$.
	We define two one-to-one mappings $\pi_1:|\A|\times|\Theta|\to[|\A|\cdot|\Theta|]$ and $\pi_2:|\cR|\times|A|\times|\Theta|\to[|\cR|\cdot|A|\cdot|\Theta|]$ such that $\pi_1(\cdot)$ associate every tuple of actions $\avec$ and state of nature $\theta$ to the index $\pi_1(\avec,\theta)$, while $\pi_2(\cdot)$ associate every receiver $r$, action $ a\in A$ and state of nature $\theta$ to the index $\pi_2(r, a, \theta)$. 
	%Formally $\pi_1(\avec, \theta)=i$ and $\pi_2(r, \hat a, \theta)=j$.
	Then we can define $i:=\pi_1(\avec,\theta)$ and $j:=\pi_2(r, a, \theta^\prime)$ so that:
	\begin{itemize}
		\item $\xvec[i]:=\phi_\theta(\avec)$
		\item $\vc[i]:=\mu_\theta\cdot \us(\avec,\theta)$
		\item $\vb[j]:=\varphi^{r,k_r}_\theta( a)$
		\item $\mA[j,i]:=\mathbb{I}(a_r= a,\theta=\theta^\prime)$.
	\end{itemize}
	
	Then we can write LP~\ref{lp:multi} as $\max_{\xvec}\vc^\top \xvec$ subject to $\xvec\ge 0$ and $\mA\xvec=\vb$. We note that the variables $\varphi$ only appear in the right-hand side vector $\vb$ in the standard-form linear program above.
\end{proof}

\lemmaConcave*

\begin{proof}
	Let $\kvec \in \K$ be a tuple of types. 
	Lemma~\ref{lem:bersimas} relates the solution $g^\kvec(\varphi)$ of LP~\ref{lp:multi} to the solution of a standard-form linear program in which $\varphi$ is the right-hand side vector of an equality constraint. Thus, for every fixed $\kvec$, the function $g^\kvec(\varphi)$ is known to be concave in $\varphi$~\citep[Theorem~5.1]{bertsimas1997introduction}.
	%The LP~\ref{lp:multi} can depends on the marginals $\{\varphi^{r,k_r}\}_{r\in\R}$ that constraint the signaling scheme to be have marginals $\{\varphi^{r,k_r}\}_{r\in\R}$, thus $\varphi$ only appears in the right hand side ‘‘right-hand side'' vector of LP~\ref{lp:multi}. Thus, $g^\kvec(\varphi)$, it is know to be a concave function of $\varphi$~\citep[Theorem~5.1]{bertsimas1997introduction}.
\end{proof}

\lemmaCont*

\begin{proof}
	First we note that for any fixed tuple of types $\kvec$, the menus $\varphi^{r,k}_\theta$ for $k\neq k_r$ do not appear, thus, in this proof, we can ease the notion by dropping $k_r$ from $\varphi^{r,k_r}_\theta$, which will be denoted by just $\varphi^{r}_\theta$.
	
	Then, for ease of clarity, we define 
	\[
	o^\kvec(\phi):=\sum_{\theta \in \Theta} \sum_{\va \in \A} \mu_\theta \phi_\theta(\va) \us(\va,\theta),
	\]
	and
	\[
	\cM^\kvec(\varphi):=\left\{\phi\,\Bigg\vert\,\sum_{\substack{\va \in \A:\\ a_r\in \hat a}} \phi_\theta(\va)  = \varphi^{r}_\theta(\hat a),\, \forall r \in \rec, \hat a \in \A_r, \theta \in \Theta\right\},
	\]
	which are the objective function and the constraints polytope of LP~\ref{lp:multi}, respectively. Formally, it holds that 
	\[
	g^\kvec(\varphi)=\max\limits_{\phi\in\cM^\kvec(\varphi)}o^\kvec(\phi).
	\]
	We will also use the function $\pi_2:\cR\times\A\times\Theta\to[|\cR|\cdot|\A|\cdot|\Theta|]$ introduced in Lemma~\ref{lem:bersimas}, that associate for every $(\hat r,\hat a,\hat \theta)\in\cR\times \A_r\times\Theta$ an index $i=\pi_2(\hat r,\hat a,\hat \theta)$.
	We first prove the $1$-Lipschitzness of $g^\kvec$ w.r.t. to $\|\cdot\|_1$. Consider any two $\varphi,\overline\varphi\in\Lambda$. 
	
	Let then $\phi\in\arg\max_{\phi^\prime\in\cM^\kvec(\varphi)}o^\kvec(\phi^\prime)$ and $\overline\phi\in\arg\max_{\phi^\prime\in\cM^\kvec(\overline\varphi)}o^\kvec(\phi^\prime)$ the values of the solutions of LP~\ref{lp:multi} w.r.t. $\varphi$ and $\overline\varphi$, respectively. 
	
	The idea of the proof is to construct a new variable $\varphi^\star$ and $\phi^\star$ that satisfies the following conditions:
	\begin{enumerate}
		\item $\phi^\star\in\cM^\kvec(\phi^\star)$
		\item $0\preceq\varphi^\star\preceq\overline\varphi$, which has to be interpreted element-wise.
		\item $\|\varphi-\overline\varphi\|_1+o^\kvec(\phi^\star)\ge o^\kvec(\phi)$.
	\end{enumerate}	
	Note that we do not require that $\varphi^\star\in\Lambda$.
	Assume that we can have such a $\varphi^\star$ and $\phi^\star$ then we can easily prove $1$-Lipschitzness w.r.t. $\|\cdot\|_1$ as follows:
	\begin{align*}
		g^\kvec(\overline\varphi)&\ge o^\kvec\left(\overline\phi\right)\\
		&\ge o^\kvec(\phi^\star)\\
		&\ge o^\kvec(\phi)-\|\varphi-\overline\varphi\|_1\\
		&=g^\kvec(\varphi)-\|\varphi-\overline\varphi\|_1,
	\end{align*}
	where the first inequality holds since $\overline\varphi-\varphi^\star\succeq0$ by assumption and thus $\overline\phi\succeq\phi^\star$ which implies that $o^\kvec(\overline\phi)\ge o^\kvec(\phi^\star)$, and the second inequality holds by assumption on $\varphi^\star$.
	This in turn implies that $|g^\kvec(\overline\varphi)-g^\kvec(\varphi)|\le\|\varphi-\overline\varphi\|_1$ since the construction is symmetric w.r.t. $\varphi$ and $\overline\varphi$.
	After we prove that $|g^\kvec(\overline\varphi)-g^\kvec(\varphi)|\le\|\varphi-\overline\varphi\|_1$ we can easily conclude the proof by observing that $\|\varphi-\overline\varphi\|_1\le\sqrt{nd|\A_r|}\cdot\|\varphi-\overline\varphi\|_2$.

	Now we show the existence such a $\varphi^\star$ and the related $\phi^\star\in\cM^\kvec(\varphi^\star)$ by explicitly building it iteratively as follows. The procedure above maintains variables $(\varphi^t,\phi^t)$ that is updated as detailed in Algorithm~\ref{alg:iterative}.
	
	\begin{algorithm}[H]\caption{\textsc{}}\label{alg:iterative}
		\begin{algorithmic}[1]
			\STATE $\varphi^0\gets\varphi$
			\STATE $\phi^0\gets\phi$
			\STATE $T\gets |\cR|\cdot|\A_r|\cdot  |\Theta|$
			\STATE $\tilde\varphi\gets\min(\overline\varphi,\varphi)$
			%\STATE $\delta_t:=\varphi^{t,r}_\theta(\hat a)-\tilde\varphi_\theta^r(\hat a)$ for $t=\pi_2(r,\theta,\hat a)$ and $r\in\R,\hat a\in\A_t, \theta\in\Theta$
			%\REQUIRE  any set of marginal signaling schemes $\varphi_1 \in \Lambda$, learning rate $\alpha$ 
			\FOR{ $t=1$ to $T$}
			\STATE $(\hat r,\hat a, \hat \theta)\gets\pi_2^{-1}(t)$
			\STATE $\delta_t\gets\varphi_{\hat \theta}^{t-1,\hat r}(\hat a)-\tilde\varphi_{\hat \theta}^{\hat r}(\hat a)$
			\STATE $\varphi^{t}\gets\varphi^{t-1}$
			\STATE$\phi^{t}\gets\phi^{t-1}$
			%				\IF{$\varphi_\theta^{t-1,r}(\hat a)\le\tilde\varphi_\theta^r(\hat a)$}
			%					\STATE $\varphi^{t,r}_\theta(\hat a)\gets \varphi^{t-1,r}_\theta(\hat a),\,\forall r\in\cR,\theta\in\Theta,a\in\A_r$
			%					\STATE$\phi^{t}_\theta(\avec)\gets \phi^{t-1}_\theta(\avec),\,\forall \avec\in\A$
			\IF{$\varphi_{\hat\theta}^{t-1,\hat r}(\hat a)\ge\tilde\varphi_{\hat \theta}^{\hat r}(\hat a)$}
			
			\STATE $\varphi^{t,\hat r}_{\hat \theta}(\hat a)\gets\tilde\varphi_{\hat \theta}^{\hat r}(\hat a)$
			\STATE $\varphi^{t,r^\prime}_{\hat \theta}(a^\prime)\gets\varphi_{\hat \theta}^{t-1,\hat r}( a^\prime)-\frac{\delta_t}{\varphi^{t-1,\hat r}_{\hat \theta}(\hat a)}\sum_{\avec\in\A:a_r=\hat a, a_{r^\prime}=a^\prime}\phi_{\hat \theta}^{t-1}(\avec),\,\forall r^\prime\neq \hat r,a^\prime\in\A_{r^\prime}$
			\STATE $\phi_{\hat \theta}^t(\avec)\gets \phi_{\hat \theta}^{t-1}(\avec)\left(1-\frac{\delta_t}{\varphi_{\hat \theta}^{t-1,\hat r}(\hat a)}\right),\,\forall\avec: a_r=\hat a$
			\ENDIF
			\ENDFOR
			\STATE\textbf{return} $\varphi^\star:=\varphi^T,\phi^\star:=\phi^T$
		\end{algorithmic}
	\end{algorithm}
	
	The idea of the procedure in Algorithm~\ref{alg:iterative}, is to maintain the constraints $\phi^t\in\cM^\kvec(\varphi^t)$ valid trough tout the procedure, and to update $\phi^t$ as to guarantee that $o^\kvec(\phi^t)\ge o^\kvec(\phi^{t-1})-\delta_t$.
	
	Now we see that the constraints $\phi^t\in\cM^\kvec(\varphi^t)$ are maintained at iteration $t$, assuming that are satisfied at time $t-1$. 
	
	Define $(\hat r,\hat \theta,\hat a)=\pi_2^{-1}(t)$ and consider the following two cases:

	• If $\varphi_{\hat \theta}^{t-1,\hat r}(\hat a)\le\tilde\varphi_{\hat \theta}^{\hat r}(\hat a)$:
	
	Then we trivially have that $\phi^t\in\cM^\kvec(\varphi^t)$ as $\phi^t=\phi^{t-1}$ and $\varphi^t=\varphi^{t-1}$ and $\phi^{t-1}\in\cM^\kvec(\varphi^{t-1})$ by assumption.
	
	• If otherwise $\varphi_{\hat \theta}^{t-1,\hat r}(\hat a)\ge\tilde\varphi_{\hat \theta}^{\hat r}(\hat a)$. We can divide the variables $(r,a,\theta)\in\cR\times\A_r\times \Theta$ into three sets
	\begin{enumerate}[label=\alph*)]
		\item $A_1:=\{(r,\theta,\hat a)\}$
		\item $A_2:=\{(r,\theta,a):a\in\A_r, a\neq\hat a\}$
		\item $A_3:=\{(r^\prime,a^\prime,\hat \theta):r^\prime\in\cR/\{\hat r\}, a^\prime\in\A_{r^\prime}\}$
		\item $A_4:=\{(r,a,\theta^\prime):\theta^\prime\in\Theta,\theta^\prime\neq\hat\theta\}$
	\end{enumerate}
	Notice that these sets are disjoint and their union is $\cR\times\A_r\times \Theta$.
	
	\textbf{a)} For any $(r,a,\theta)\in A_1$ we have:
	\begin{align*}
		\sum\limits_{\avec\in\A:a_r= a}\phi^t_\theta(\avec)&=\sum\limits_{\avec\in\A:a_{ r}= a} \phi_{ \theta}^{t-1}(\avec)\left(1-\frac{\delta_t}{\varphi_{ \theta}^{t-1, r}( a)}\right)\\
		&=\varphi_\theta^{t-1}( a)\left(1-\frac{\delta_t}{\varphi_\theta^{t-1}( a)}\right)\\
		&=\varphi^{t-1}_\theta( a)-\delta_{t}\\
		&=\tilde\varphi_\theta^{r}( a).
	\end{align*}
	
	\textbf{b)} For any $(r,a,\theta)\in A_2$ we have:
	\begin{align*}
		\sum\limits_{\avec\in\A:a_r= a^\prime}\phi^t_\theta(\avec)&=\sum\limits_{\avec\in\A:a_r= a^\prime}\phi^{t-1}_\theta(\avec)=\varphi^{t-1,r}_\theta(a^\prime)=\varphi_\theta^{t,r}(a^\prime).
	\end{align*}
	as the those variable are not updated at round $t$.

	\textbf{c)}  For any $(r,a,\theta)\in A_3$ we have:
	\begin{align*}
		\sum\limits_{\avec\in\A:a_{r}=a}\phi^{t}_{\theta}(\avec)&=\sum\limits_{\substack{\avec\in\A:\\a_{ r}=a\\
				a_{\hat r }=\hat a}}\phi^{t}_{\theta}(\avec)+\sum\limits_{\substack{\avec\in\A:\\a_{r}=a\\
				a_{\hat r}\neq\hat a}}\phi^{t}_{\theta}(\avec)\\	
		&=\sum\limits_{\substack{\avec\in\A:\\a_{r}^{}=a\\
				a_{\hat r}=\hat a}}\phi^{t-1}_{\theta}(\avec)\left(1-\frac{\delta_t}{\varphi_\theta^{t-1,\hat r}(\hat a)}\right)+\sum\limits_{\substack{\avec\in\A:\\a_{r}=a\\
				a_{\hat r}\neq\hat a}}\phi^{t-1}_{\theta}(\avec)\\	
		%	&=\sum\limits_{\substack{\avec\in\A:\\a_{r^\prime}^{\prime}=a^\prime\\
				%		a_r=\hat a}}\phi^{t-1}_{\theta^\prime}(\avec)\left(1-\frac{\delta_t}{\varphi_\theta^{t-1,r}(\hat a)}\right)+\sum\limits_{\substack{\avec\in\A:\\a_{r^\prime}^{\prime}=a^\prime\\
				%		a_r\neq\hat a}}\phi^{t-1}_{\theta^\prime}(\avec)\\	
		&=\sum\limits_{\substack{\avec\in\A:\\a_{r}^{}=a}}\phi^{t-1}_{\theta^\prime}(\avec)-\frac{\delta_t}{\varphi_\theta^{t-1,\hat r}(\hat a)}\sum\limits_{\substack{\avec\in\A:\\a_{r}^{}=a\\
				a_{\hat r}=\hat a}}\phi^{t-1}_{\theta}(\avec)\\
		&=\varphi_{\theta}^{t,r}(a),
	\end{align*}
	where for the second equality we used the update of update of $\phi^{t-1}(\avec)$ in Line 13 of Algorithm~\ref{alg:iterative}. While the last equality follows from the update of Line 12.
	
	\textbf{d)}  For any $(r,a,\theta)\in A_4$ we have that none of the variable are updated an thus the statement holds by inductive assumption.
	
	This proves that $\phi^\star\in\cM^\kvec(\varphi^\star)$.
	
	On the other hand it is evident that $\varphi^\star\preceq \overline\varphi$ thanks to update of Line 11 in Algorithm~\ref{alg:iterative}. In particular it also holds that $\varphi_{ \hat\theta}^{t,\hat r}(\hat a)\le\overline\varphi_{ \hat\theta}^{\hat r}(\hat a)$ for all $t=\pi_2(\hat r,\hat a,\hat \theta)$.
	
	We are left to show that $\|\varphi-\overline\varphi\|_1+o^\kvec(\phi^\star)\ge o^\kvec(\phi)$. Consider the following inequalities:
	\begin{align*}
		o^\kvec(\phi^t)&:=\sum_{\theta \in \Theta} \sum_{\va \in \A} \mu_\theta \phi^t_\theta(\va) \us(\va,\theta)\\
		&=\mu_{\hat \theta} \sum_{\va \in \A}  \phi^t_{\hat \theta}(\va) \us(\va,\hat \theta)+\sum_{\substack{\theta \in \Theta/\{\hat\theta\}}} \sum_{\va \in \A} \mu_\theta \phi^t_\theta(\va) \us(\va,\theta)\\
		&=\mu_{\hat \theta} \sum_{\va \in \A}  \phi^t_{\hat \theta}(\va) \us(\va,\hat \theta)+\sum_{\substack{\theta \in \Theta/\{\hat\theta\}}} \sum_{\va \in \A} \mu_\theta \phi^{t-1}_\theta(\va) \us(\va,\theta)\\
		&=\mu_{\hat \theta} \sum_{\substack{\va \in \A:\\a_{\hat r}=\hat a}}  \phi^t_{\hat \theta}(\va) \us(\va,\hat \theta)+\mu_{\hat \theta} \sum_{\substack{\va \in \A:\\a_{\hat r}\neq\hat a}}  \phi^t_{\hat \theta}(\va) \us(\va,\hat \theta)+\sum_{\substack{\theta \in \Theta/\{\hat\theta\}}} \sum_{\va \in \A} \mu_\theta \phi^{t-1}_\theta(\va) \us(\va,\theta)\\
		&=\mu_{\hat \theta}\left(1-\frac{\delta_t}{\varphi_{\hat\theta}^{t-1,\hat r}(\hat a)}\right) \sum_{\substack{\va \in \A:\\a_{\hat r}=\hat a}}  \phi^{t-1}_{\hat \theta}(\va) \us(\va,\hat \theta)+\mu_{\hat \theta} \sum_{\substack{\va \in \A:\\a_{\hat r}\neq\hat a}}  \phi^{t-1}_{\hat \theta}(\va) \us(\va,\hat \theta)+\sum_{\substack{\theta \in \Theta/\{\hat\theta\}}} \sum_{\va \in \A} \mu_\theta \phi^{t-1}_\theta(\va) \us(\va,\theta)\\
		&=\sum_{\theta \in \Theta} \sum_{\va \in \A} \mu_\theta \phi^{t-1}_\theta(\va) \us(\va,\theta)-\frac{\delta_t}{\varphi_{\hat\theta}^{t-1,\hat r}(\hat a)}\mu_{\hat\theta}\sum_{\substack{\va \in \A:\\a_{\hat r}=\hat a}}  \phi^{t-1}_{\hat \theta}(\va) \us(\va,\hat \theta)\\
		&\ge o^\kvec(\phi_\theta^{t-1})-\frac{\delta_t}{\varphi_{\hat\theta}^{t-1,\hat r}(\hat a)}\sum\limits_{\substack{\avec\in\A:\\a_{\hat r}=\hat a}}\phi_{\hat\theta}^{t-1}(\avec)\\
		&=o^{\kvec}(\phi_\theta^{t-1})-\delta_t.
	\end{align*}
	Then we can telescope the inequality to show that:
	\[
	o^\kvec(\phi^\star)\ge o^\kvec(\phi)-\sum\limits_{t=1}^T\delta_t.
	\]
	
	Then it is easy to show that $\delta_t=\varphi_{\hat \theta}^{t-1,\hat r}(\hat a)-\tilde\varphi_{\hat \theta}^{\hat r}(\hat a)\le \overline\varphi_{\hat \theta}^{\hat r}(\hat a)-\tilde\varphi^{\hat r}_{\hat \theta}(\hat a)\le | \overline\varphi_{\hat \theta}^{\hat r}(\hat a)-\tilde\varphi^{\hat r}_{\hat \theta}(\hat a)|$ and thus 
	\[
	o^\kvec(\phi^\star)\ge o^\kvec(\phi)-\|\varphi-\overline\varphi\|_1,
	\]
	as wanted.
\end{proof}

\regretMultiType*

\begin{proof}
	First notice that:
	\[
	\max\limits_{\varphi\in\Lambda}\sum\limits_{t=1}^Tg^{\kvec_t}(\varphi)=\max\limits_{\varphi\in\Lambda}\sum\limits_{t=1}^T\us(\varphi,\kvec_t)
	\]
	which follows from the definition of $\us(\varphi,\kvec)$ given in Equation~\eqref{eq:SenderUtilityOnMenu}.
	On the other hand it is clear that $g^{\kvec_t}(\varphi_t)=\us(\phi_t,\kvec_t)$ thanks to the update of Line~4 of Algorithm~\ref{alg:typeReporting}.
	
	Thus we can write the regret of Algorithm~\ref{alg:typeReporting} as:
	\[
	R_T=\max\limits_{\varphi\in\Lambda}\sum\limits_{t=1}^Tg^{\kvec_t}(\varphi)-\sum\limits_{t=1}^Tg^{\kvec_t}(\varphi_t).
	\] 
	Then, by Lemma~\ref{lem:concave} we have that the reward functions $g^\kvec(\cdot)$ are concave for all $\kvec \in \K$.
	
	Moreover we know that the by Lemma~\ref{lm:cont} that for all $\kvec \in \K$ the functions $g^\kvec(\cdot)$, are $\sqrt{nd|A|}$-Lipschitz w.r.t. $\|\cdot\|_2$ and thus, by~\citet[Lemma~2.6]{shalev2012online}, we have that all the subgradients of $-g^\kvec(\cdot)$ have norm bounded by the Lipschitz constant. This clearly implies $G:=\sup_{\varphi\in\Lambda}\|\partial g^\kvec(\varphi)\|_2\le\sqrt{nd|A|}$.
	
	Moreover, the regularizer $\frac{1}{2}\|\cdot\|_2^2$ is trivially $1$-strongly convex w.r.t. $\|\cdot\|_2$. 
	
	Finally we have that the diameter of the polytope $\Lambda$, induced by the regularizers is bounded by $\frac{1}{2}ndm|A|$, as $\Lambda$ is a contained in the $ndm|A|$-dimensional hypercube. Formally $D:=\sqrt{\max_{\phi \in \Lambda} \frac{1}{2}\lVert\phi\rVert_2^2-\min_{\phi^\prime\in\Lambda}\frac{1}{2}\lVert\phi'\rVert_2^2}\le \sqrt{\frac{1}{2} ndm|A|}$.
	
	A standard application of~\citet[Corollary~7.9]{Orabona} gives a bound of:
	\[
	R_T\le \frac{D^2}{\alpha}+\frac{1}{2}\alpha G^2 T\le \frac{1}{2\alpha}ndm|A|+\frac{1}{2}\alpha nd|A| T.
	\]
	Setting $\alpha=\sqrt{m/T}$ gives the result.
	%
	%The decision space $\Lambda$ is a polytope with diameter $D=\sqrt{\max_{\phi \in \Lambda} \frac{1}{2}\lVert\phi\rVert_2^2-\min_{\phi^\prime\in\Lambda}\frac{1}{2}\lVert\phi'\rVert_2^2}\le \frac{1}{2} ndm|A|$ since $\Lambda$ is a subset of the $ndm|\A_r|$ dimensional hypercube. Moreover, all the reward functions $g^\kvec(\cdot)$, $\kvec \in \K$ are concave by Lemma~\ref{lem:concave}. Finally, the subgradients of $g$ has norm at most $G=||g||_2\le \sqrt{nd|A_r|}$. Finally, the regularizer $\frac{1}{2}|| ||$ is 1-strongly convex. Hence, a standard application of FTRL (see, \emph{e.g.}, \cite{Orabona}) provides regret at most $D^2/\alpha+\frac{1}{2}G^2 \alpha=nd|\A_r|\sqrt{mT}$.
\end{proof}

\section{Proofs Omitted from Section~\ref{sec:multiTypeRepEfficient}}\label{app:MultiTypeRepEfficient}

\ellipsoidEasy*

\begin{proof}
	We defining for any $R\subset\rec$, $\avec_R$ as the tuple in which action $a_1$ is recommended to all the receivers in $R$ and $a_0$ to the others. Formally $a_r=a_1$ for all $r\in R$, and $a_r=a_0$ for all $r\in \rec/ R$. Then , rewriting LP~\ref{lp:multi} for the specific case of binary actions per receiver, we obtain:
	\begin{subequations}
		\begin{align}
			\max_{\phi \ge 0}  &\sum_{\theta \in \Theta} \sum_{R \subseteq \rec} \mu_\theta\phi_\theta(\avec_R) \fs_\theta(R) \quad \text{s.t.} \label{obj:lpMultiTwo} \\
			& \sum_{\substack{R \in \rec:r \in R}} \phi_\theta(\avec_R)  = \varphi^{r,k_r}_\theta(a_1),\quad  \forall r \in \rec, \forall \theta \in \Theta \label{lp:multiTwo1} \\
			&\sum_{R \in \rec} \phi_\theta(\avec_R)  = 1,\quad \forall \theta \in \Theta \label{lp:multiTwo2} 
		\end{align}
	\end{subequations}
	%
	%
	%	The dual of such LP reads as follows, where the variables are $\{x_{r,\theta}\}_{r \in \rec,\theta \in \Theta}$ and $\{y_\theta\}_{\theta \in \Theta}$.
	%	%
	%	\begin{align*}
		%	\min_{x,y} &  \sum_{r \in \rec, \theta \in \Theta} \varphi^{r,k_r}_\theta(a_1) x_{r,\theta} + \sum_{\theta}y_\theta \quad \textnormal{s.t.}\\
		%	 &\sum_{r\in R} x_{r,\theta} + y_{\theta} \ge \mu_\theta \fs_{\theta}(R),\, \forall R\subseteq \rec,\theta \in \Theta
		%	\end{align*}
	%	%
	%	A separation oracle for dual problem can be implemented efficiently applying the optimization oracle $\mathcal{O}$.
	%	Given a assignment $(x,y)$, for each $\theta \in \Theta$ the separation oracle solves: $\max_{R \subseteq \rec}\left\{ \mu_\theta \fs_\theta(R)-\sum_{r\in R} x_{r,\theta}\right\}$. If at least one time the value is larger that $y_\theta$ it returns the violated constraints. Otherwise, the assignment is feasible.
	%
	The dual of such LP reads as follows:
	\begin{align*}
		\min_{x} &  \sum_{r \in \rec, \theta \in \Theta} \varphi^{r,k_r}_\theta(a_1) x_{r,\theta}\quad \textnormal{s.t.}\\
		&\sum_{r\in R} x_{r,\theta} \ge \mu_\theta \fs_{\theta}(R),\quad \forall R\subseteq \rec,\theta \in \Theta,
	\end{align*}
	where the dual variables are $\{x_{r,\theta}\}_{r \in \rec,\theta \in \Theta}$.
	A separation oracle for dual problem can be implemented exploiting the optimization oracle $\mathcal{O}(\fs_\theta, -x_{\theta}/\mu_\theta)$ for each $\theta$.
	If,for at least one $\theta$, the value of $\mathcal{O}(\fs_\theta, -x_{\theta}/\mu_\theta)$ is larger that $0$ then we can use the the violated constraint as a separating hyperplane.
\end{proof}

\quadratic*

\begin{proof}

	We defining for any $R\subset\rec$, $\avec_R$ as the tuple in which action $a_1$ is recommended to all the receivers in $R$ and $a_0$ to the others. Formally $a_r=a_1$ for all $r\in R$, and $a_r=a_0$ for all $r\in \rec/ R$.
	With this definition, for any sequence of type's tuples, the problem $ \max_{\varphi \in \Lambda}\sum_{\tau \in [t]} g^{\kvec_\tau}(\varphi)-\frac{1}{2\alpha} \lVert\varphi\rVert_2^2$ can be rewritten as:
	\begin{subequations}\label{lp:multiBayes}
		\begin{align}
			\max_{\phi \ge 0, \varphi \in \Lambda}   \sum_{	\substack{\tau \in [t]\\ \theta \in \Theta\\R \subseteq \rec}} &\mu_\theta  \phi_{\tau,\theta}(\avec_R) \fs_\theta(R) - \frac{1}{2\alpha} \sum_{ \substack{r \in \rec,k \in \K_r,\\\theta \in \Theta, a \in A}} \varphi^{r,k}_\theta(a)^2 \quad \text{s.t.} \\
			& \hspace{-1cm} \sum_{\substack{R \subseteq \rec:\\r \in R}} \phi_{\tau,\theta}(\avec_R)  = \varphi^{r,k_{\tau,r}}_\theta(a_1),\quad \forall r \in \rec, \theta \in \Theta, \tau \in [t] \\
			&  \hspace{-1cm} \sum_{R\subseteq \rec} \phi_{\tau,\theta}(\avec_R) = 1,\quad  \forall \tau \in [t] ,\theta \in \Theta\label{lp:multiBayes2}
		\end{align}
	\end{subequations}
	We Lagrangyfing Problem~\eqref{lp:multiBayes} by introducing the following dual variables
	\begin{itemize}
		\item $x_{r,\theta,\tau} \in \mathbb{R}$ for each $r\in\rec,\theta \in \theta$, $\tau\in [t]$, which is the dual variable of the constrain~\ref{lp:multi1}
		\item   $y_{\tau,\theta} \in \mathbb{R} $ for each $\tau \in [t]$, $\theta \in \Theta$, which is the dual variable of the constrain~\ref{lp:multiBayes2}
		\item $z_{r,k,k'} \in\mathbb{R}_+$ for each $r \in \rec,k,k' \in \K_r$, which is the dual variable of the constrain~\ref{eq:lambdaTRsingle1}
		\item $\alpha_{r,k,k',a,a'} \in \mathbb{R}_+$ for each $r \in \rec,k, k^\prime \in \K_r,a,a^\prime\in A$, which is the dual variable of the constrain~\ref{eq:lambdaTRsingle2}
		\item  $\beta_{r,k,\theta} \in \mathbb{R}  $  for each $r \in \rec,  k \in \K_r, \theta \in \Theta$, which is the dual variable of the constrain~\ref{eq:lambdaTRsingle3}
		\item $\gamma_{\theta,R}  \in \mathbb{R}_+  $ for each $\theta \in \Theta$, $R \subseteq \rec$, for the constraint $\phi\ge 0$
		\item $\eta_{r,k,\theta,a} \in \mathbb{R}_+$ for each $r \in \rec$, $k \in \K_r$, $\theta \in \Theta$, and $a \in A$, for the constraint $\varphi\ge 0$
	\end{itemize}
	The the Lagrangian of Problem~\ref{lp:multiBayes} reads:
	\begin{align*}
		L(\phi,\varphi, x,y,z,\alpha,\beta,\gamma,\eta)
		&= \sum_{	\substack{\tau \in [t], \theta \in \Theta\\R \subseteq \rec}} \mu_\theta  \phi_{\tau,\theta}(\avec_R) \fs_\theta(R) - \frac{1}{2\alpha} \sum_{ \substack{r \in \rec,k \in \K_r,\\\theta \in \Theta, a \in A}} \varphi^{r,k}_\theta(a)^2\\
		&+\sum_{\substack{ \tau \in [t], \theta\in\Theta\\r \in \rec}} x_{r,\theta,\tau} \left(  \sum_{\substack{R \subseteq \rec:\\r \in R}} \phi_{\tau,\theta}(\avec_R)  - \varphi^{r,k_{\tau,r}}_\theta(a_1)\right)\\
		&+ \sum_{\substack{\tau \in [t], \\\theta \in \Theta}} y_{\tau,\theta} \left(\sum_{R \subseteq \rec} \phi_{\tau,\theta}(\avec_R)-1\right)\\ 
		&+ \sum_{\substack{r \in \rec, \\k,k' \in \K_r }} z_{r,k,k'} \left(\sum_{a \in A} \sum_{\theta \in \Theta} \mu_\theta \, \varphi_\theta^{r,k}(a) \, \ur[r][k](a,\theta) -\sum\limits_{a\in A}l_a^{r,k,k^\prime}\right)\\
		&+ \sum_{\substack{r \in \rec, k,k' \in \K_r,\\ a,a' \in A}} \alpha_{r,k,k',a,a'}\left(l_a^{r,k,k^\prime}-  \sum_{\theta \in \Theta} \mu_\theta \, \varphi_\theta^{r,k'}(a) \, \ur[r][k](a^\prime,\theta) \right)\\
		&+ \sum_{\substack{r \in \rec,\\ k \in \K_r, \theta \in \Theta}} \beta_{r,k,\theta} \left(\sum\limits_{a\in A}\varphi^{r,k}_\theta(a)-1 \right) + \sum_{\substack{\theta \in \Theta,\\ R \subseteq \rec}} \gamma_{\theta,R} \phi_\theta(\avec_R) + \sum_{\substack{r \in \rec, k \in \K_r, \\ \theta \in \Theta, a \in A}} \eta_{r,k,\theta,a} \varphi^{r,k}_{\theta}(a).
	\end{align*}
	
	We observe that Slater's condition holds for Problem~\ref{lp:multiBayes}. This holds since all constraints are linear and there exists a feasible solution.
	This is easily seen as there exists a set of feasible menu of IC marginal signaling schemes. Moreover, given a set of menus and a vector of types, it is possible to design consistent signaling schemes by taking the product distribution of the marginal signaling schemes relative to the types. Therefore, by strong duality, the optimal primal and dual variables must satisfy the KKT conditions. In particular it must hold that $\boldsymbol{0}\in\partial_{\phi_{\tau,\theta}(\avec_R)}(L)$ for each $\tau \in [t]$, $\theta\in \Theta$, $R \subseteq \rec$. Formally, for each $\tau \in [t]$,  $\theta \in \Theta$, and $R\subseteq \rec$, we have:
	\begin{align}\label{eq:slat1}
		\partial_{\phi_{\tau,\theta}(\avec_R)}(L)=  \mu_\theta  \fs_\theta(R)+ \sum_{ r \in R }  x_{r,\theta,\tau}  + y_{\tau,\theta}+\gamma_{\theta,R}=0.
	\end{align}
	Moreover, it must also hold that $\boldsymbol{0}\in\partial_{\varphi^{r,k}_\theta(a)}(L)$. Formally, for each $r \in \rec$, $k \in \K_r$, $\theta \in \Theta$, and $a \in A$ it holds
	\[
	 - \frac{\varphi^{r,k}_{\theta}(a)}{\alpha} -  \mathbb{I}_{a=a_1} \sum_{\substack{\tau \in [t]:\\{k}=k_{\tau,r}}} x_{r,\theta,\tau}+ \left(\sum\limits_{\substack{k^\prime\in\K_r}}z_{r,k,k'}\right) \mu_\theta u^r_k(a,\theta)   - \sum_{\substack{a'\in A,\\k'\in \K_r}}  \alpha_{r,k',k,a,a'} \mu_\theta \ur[r][k'](a',\theta) +\beta_{r,k,\theta}+\eta_{r,k,\theta,a}=0, 
	\]
	which implies that for each $r \in \rec$, $k \in \K_r$, $\theta \in \Theta$, and $a \in A$:
	\begin{equation}\label{eq:slat2}
		\frac{\varphi^{r,k}_\theta(a)}{\alpha}= -  \mathbb{I}_{a=a_1} \sum_{\substack{\tau \in [t]:\\{k}=k_{\tau,r}}} x_{r,\theta,\tau}+ \left(\sum\limits_{\substack{ k^\prime\in\K_r}}z_{r,k,k'}\right) \mu_\theta u^r_k(a,\theta)   - \sum_{\substack{a'\in A,\\k'\in \K_r}}  \alpha_{r,k',k,a,a'} \mu_\theta \ur[r][k'](a',\theta) +\beta_{r,k,\theta}+\eta_{r,k,\theta,a}.
	\end{equation}
	Similarly, it must hold that $\boldsymbol{0}\in \partial_{l^{r,k,k'}_a}(L)$. Formally, for each $r \in \rec$, $k,k' \in \K_r$, and $a \in A$, it holds
	\begin{equation}\label{eq:slat3}
		\partial_{l^{r,k,k'}_a}(L) = - z_{r,k,k'} +\sum_{a' \in A} \alpha_{r,k,k',a,a'} =0 
	\end{equation}
	Finally, plugging Equation~\eqref{eq:slat1}, Equation~\eqref{eq:slat2} and Equation~\eqref{eq:slat3} back into the Lagrangian we get:
	\begin{align*}\label{eq:slat4}
		L(\phi,\varphi,x,y,z,\alpha,\beta,\gamma,\eta)&
		=  \frac{1}{2\alpha}\sum_{ \substack{r \in \rec,k \in \K_r,\\\theta \in \Theta, a \in A}} \varphi^{r,k}_\theta(a)^2-\sum_{\tau \in [t], \theta \in \Theta} y_{\tau,\theta} - \sum_{r \in \rec, k \in \K_r, \theta \in \Theta} \beta_{r,k,\theta}.
	\end{align*}
	Finally the dual problem of Problem~\ref{lp:multiBayes} can be written as follows:
	\begin{subequations}\label{prob:dualEfficientMulti}
		\begin{align}
			\min_{\varphi, x,\beta\le 0} &\left\{\frac{1}{2\alpha}\sum_{ \substack{r \in \rec,k \in \K_r,\\\theta \in \Theta, a \in A}} \varphi^{r,k}_\theta(a)^2-\sum_{\tau \in [t], \theta \in \Theta} y_{\tau,\theta} - \sum_{r \in \rec, k \in \K_r, \theta \in \Theta} \beta_{r,k,\theta}\right\}		\quad\textnormal{ s.t. } \\
			&  \mu_\theta  \fs_\theta(R)+ \sum_{ r \in R }  x_{r,\theta,\tau}  + y_{\tau,\theta} \le 0,\quad \forall \tau \in [t],  \theta \in \Theta,  R \subseteq \rec \label{eq:neg} \\
			&\frac{\varphi^{r,k}_\theta(a)}{\alpha} \le-  \mathbb{I}_{a=a_1} \hspace{-0.2cm}\sum_{\substack{\tau \in [t]:\\{k}=k_{\tau,r}}} x_{r,\theta,\tau}+ \left(\sum\limits_{k'\in\K_r}z_{r,k,k'}\right) \mu_\theta u^r_k(a,\theta)   -\hspace{-0.2cm} \sum_{\substack{a'\in A,\\k'\in \K_r}}  \alpha_{r,k',k,a,a'} \mu_\theta \ur[r][k'](a',\theta)+\beta_{r,k,\theta,a},\nonumber\\
			&	\hspace{8cm}\forall r \in \rec, k \in \K_r, \theta \in \Theta, a \in A \label{eq:neg2}\\
			&- z_{r,k,k'} +\sum_{a' \in A} \alpha_{r,k,k',a,a'} =0,\quad \forall r\in \rec,  k,k' \in \K_r, \forall a \in A
		\end{align}
	\end{subequations} 
	where the constraint of Equation~\eqref{eq:slat1} becomes the constraint of Equation~\eqref{eq:neg} since the dual variable $\gamma$ is positive. 
	Similarly, the constraint of Equation~\eqref{eq:slat2} becomes the constraint of Equation~\eqref{eq:neg2} as the dual variable $\eta$ is positive.
	
	We now remark that the above dual problem can be solve in polynomial time, when we have access to the optimization oracle $\cO$.
	
	Problem~\ref{prob:dualEfficientMulti} is convex. Hence, we can solve it applying the ellipsoid method.
	The separation over Constraint~\eqref{eq:neg2} can be done in polynomial-time since there are polynomially-many constraints. Moreover, the separation problem relative to the objective can be solved in polynomial time since there are polynomially-many variables and the objective is convex. 
	Finally, the separation over the constraint of Equation~\eqref{eq:neg} must solve 
	\[
	\arg \max_{R}  \left\{ \mu_\theta  \fs_\theta(R) +\sum_{ r \in \rec }  x_{r,\theta,\tau}\right\},
	\]
	for each possible $\tau \in [t]$ and $\theta \in \Theta$, which can be done by exploiting the optimization oracle $\cO(\fs_{\theta}, x_{\theta,\tau}/\mu_\theta)$ for all $\tau\in[t]$ and $\theta\in\Theta$.
	
	If any of these solution are greater than $-y_{\tau,\theta}$, we return the relative constraint, otherwise all the constraints~\eqref{eq:neg2} are satisfied.
	Hence, the ellipsoid method runs in polynomial-time and find an arbitrary good approximation. 
	For the easy of exposition, we ignore the arbitrary small approximation error of the ellipsoid method.
\end{proof}

\corollaryOracle*

\begin{proof}
	Since, by Lemma~\ref{lem:dughmi}, there exists a polynomial-time oracle $\mathcal{O}$, applying Lemma~\ref{lm:oracle1} and~\ref{lem:oraclequadratic} we can compute Line~4 and~5 of Algorithm~\ref{alg:typeReporting} in polynomial-time. 
	Moreover, it is easy to see that all the other operations of the algorithm can be executed in polynomial time.
\end{proof}

\end{document}